\documentclass[usletter, 12pt]{article}

\usepackage{amsmath,amsthm,amssymb,amstext,amsfonts,mathrsfs}
\usepackage{array, makecell}
\usepackage{authblk}
\usepackage[english]{babel}
\usepackage[utf8]{inputenc}
\usepackage{caption}
\usepackage{changepage}
\usepackage[usenames,dvipsnames]{color, xcolor}
\usepackage{comment}
\usepackage{enumerate,enumitem}
\usepackage[T2A]{fontenc}
\usepackage{graphicx,pgf}
\usepackage{geometry}
\usepackage{hhline}
\usepackage[hidelinks]{hyperref}
\hypersetup{colorlinks = true, urlcolor = blue, linkcolor = blue, citecolor = blue, allcolors=blue}
\usepackage{lipsum}
\usepackage{natbib}
\usepackage{multirow,multicol}
\usepackage{pbox}
\usepackage{rotating,xspace}
\usepackage{setspace}
\usepackage{subcaption}
\usepackage{type1cm}
\usepackage{tabularx}
\usepackage{tcolorbox}
\usepackage{tikz}
\usetikzlibrary{positioning}
\usetikzlibrary{decorations.pathreplacing}
\usetikzlibrary{arrows,backgrounds}
\tikzstyle{object}=[circle,draw=red]
\tikzstyle{agent}=[circle,draw=blue]
\tikzstyle{quantity}=[fill=white]

\def\w{\omega}

\def\cv{\mathcal{V}}
\def\ce{\mathcal{E}}
\def\bx{\mathbf{x}}
\def\G{\Gamma}

\hypersetup{colorlinks = true}
\hypersetup{allcolors=blue}

\newtheorem{definition}{Definition}

\newtheorem{proposition}{Proposition}
\newtheorem{lemma}{Lemma}
\newtheorem{example}{Example}
\newtheorem{corollary}{Corollary}

\newcommand{\norm}[1]{| #1 |}

\newcommand{\hide}[1]{} 

\usepackage{setspace}
\begin{document}
	
	\title{Efficient and fair trading mechanisms for resource exchange in market design}

\author{Jingsheng Yu\thanks{School of Economics and Management, Wuhan University. Email: yujingsheng1987@outlook.com} \quad \quad Jun Zhang\thanks{Institute for Social and Economic Research, Nanjing Audit University. Email: zhangjun404@gmail.com}
}


\maketitle	

\begin{abstract}\label{abstract}

We develop a method using parameterized linear equations to define trading mechanisms in market design models. Our method adeptly addresses challenges arising from factors such as complex endowments or coarse priorities, while offering flexibility to incorporate fairness concerns through the selection of equation parameters. Applying this method to models including fractional endowment exchange, priority-based allocation, and house allocation with existing tenants, we obtain new mechanisms that are both efficient and fair in these models.

\end{abstract}	

\bigskip

\noindent \textbf{Keywords}: market design; trading mechanism; efficiency; fairness

\noindent \textbf{JEL Classification}: C71, C78, D71

\thispagestyle{empty}
\setcounter{page}{0}


\newpage

\section{Introduction}\label{section:intro}

Efficiency and fairness are two central desiderata in market design. Trading mechanisms have been widely used in various models to find efficient allocations. Whenever an allocation is inefficient, there exist opportunities for agents to exchange their assignments to enhance welfare. This paper develops a new method for defining trading mechanisms and applies it to design new mechanisms that are efficient and fair in several well-studied models.

A pioneering trading mechanism in market design is Gale's \textit{top trading cycle} (TTC) for the housing market model \citep{ShapleyScarf1974}. In this model, finite agents own distinct objects and seek to exchange their endowments. TTC iteratively constructs directed graphs in which agents point to their most preferred objects and objects point to their owners. At each step, agents forming cycles in the graph trade their endowments. This graph-based cycle-clearing idea forms the foundation for subsequent trading mechanisms in the literature.\footnote{See \cite{AbduSonmez1999,papai2000strategyproof,roth2004kidney,erdil2008s,KestenUnver2014,erdil2017two,pycia2017incentive,dur2019two}.} For instance, assuming strict priorities, \cite{abdulkasonmez2003} generalize TTC to the school choice model by allowing students to trade their priorities for schools.

The cycle-clearing definition of trading mechanisms is appealing for specifying details of the trading process, such as who trades with whom and how trades occur. However, many environments are more complex than the housing market model or the school choice model with strict priorities. For instance, when an object has multiple owners or when agents are tied in priority for an object, the graph generated by the cycle-clearing idea can be complicated, with objects pointing to several agents and agents involved in overlapping cycles. Determining how to clear these cycles becomes a nontrivial challenge, especially when randomization must be incorporated into the trading process to address fairness concerns.\footnote{When co-ownership or priority ties exist, randomly generating private endowments or strict priorities before running TTC can result in efficiency loss. In general, when agents are involved in overlapping cycles in a graph, randomly selecting cycles to clear can result in efficiency loss.}

This paper introduces a ``reduced form'' approach to defining the trading process in a given graph. We characterize the outcome of a trading process using a linear equation system that describes an input-output equilibrium for all nodes in the graph; that is, each node supplies a certain amount of resources in exchange for an equal amount from others. The system includes parameters representing how agents with rights over a resource divide the rights to use it for trading. These parameters provide flexibility in controlling the trading process and can be selected to address fairness or other concerns. Compared with the cycle-clearing idea, our method circumvents the challenge of identifying cycles and specifying how trades occur in complex environments. Moreover, computation is not a concern for our method, since linear equations can be efficiently solved by well-developed computational methods.

The equation system we employ to describe a trading process resembles the classical Leontief model \citep{leontief1941}, which describes a general equilibrium in an economy consisting of finite industries. The input-output relationship between the industries is characterized by a linear equation system with coefficients forming a stochastic matrix (i.e., each element belongs to $ [0, 1] $ and each column sums to one). This is also the key feature of the coefficients in our equation system. Results on the Leontief model (e.g. \citealp{peterson1982leontief}) ensure that our equation system has positive solutions (Lemma \ref{thm:existence}). Thus, the trading mechanisms in our approach are well defined.

To demonstrate the effectiveness of our method, we apply it to three models and derive new mechanisms for each. We first consider the \textit{fractional endowment exchange} (FEE) model, a direct generalization of the housing market model. In this model, agents may own fractional shares of multiple objects, and each object may be jointly owned by multiple agents in varying shares. The intricate endowment structure of this model presents significant challenges in designing efficient and fair trading mechanisms, making it an ideal setting to showcase the advantages of our method. The FEE model is also notable for unifying the housing market model and another classical model known as house allocation \citep{hylland1979efficient,abdulkadirouglu1998random}, in which all agents collectively own all objects. We view the house allocation model as a special case of the FEE model in which agents own equal shares of all objects. For these reasons, we devote a substantial portion of the paper to this model. Once the FEE model is analyzed, applying our method to the other two models becomes straightforward.


We propose a class of \textit{balanced trading mechanisms} (BTMs) for the FEE model. At each step of these mechanisms, agents report their most preferred objects, and an equation system specifies the trading process. All mechanisms in this class are individually rational and sd-efficient (Proposition \ref{prop:BTM:efficiency}), and they differ only in their equation parameters. The definition of BTM is intentionally broad to accommodate various selections of parameters. These selections determine the fairness and other properties of the mechanisms. We therefore devote substantial discussions to the selection of parameters. We first show that, by selecting parameters satisfying conditions we provide, these mechanisms can achieve fairness criteria including equal treatment of equals (ETE), equal-endowment no envy (EENE), and a new criterion we introduce called \textit{bounded envy} (Proposition \ref{prop:BTA:fairness}). ETE and EENE are standard axioms that ensure fairness between agents with equal endowments (and identical preferences). Bounded envy is a more general criterion applying to any pair of agents, requiring that the potential envy one agent may feel toward another be bounded by the latter's relative advantage in endowments. Bounded envy implies EENE, which implies ETE.

Since the FEE is a model of endowment exchange, and the equation parameters in BTMs determine agents' trading opportunities, it seems natural to select parameters based only on the distribution of agents' endowments. Thus, we then study a subclass of BTMs with this feature, which we call \textit{regular}. All regular BTMs satisfy EENE and are anonymous and neutral (Proposition \ref{prop:regular:fairness}), meaning that all agents and objects are treated equivalently regardless of their identities. Although individual rationality and sd-efficiency together are incompatible with strategy-proofness, we prove that regular BTMs satisfy \textit{bounded invariance}, meaning that an agent cannot manipulate the allocation of an object by changing the reported preferences over less preferred objects, and that a subclass of regular BTMs are \textit{asymptotically strategy-proof} in large economies (Proposition \ref{prop:asymptotic:IC}). In proving asymptotically strategy-proofness, the EENE property of regular BTMs plays a central role.

Regular BTMs also satisfy \textit{decomposability}: if an agent is split into sub-agents, each owning exactly one of the original agent’s endowments, then the combined assignments of these sub-agents equal the assignment of the original agent, while other agents’ assignments remain unchanged. Consequently, any economy can be decomposed into a \textit{simple} economy in which each agent owns only one object, and each pair of agents either owns distinct objects or owns the same object, but potentially in different amounts. Every regular BTM is uniquely characterized by its allocations in these simple economies.

Within the class of BTMs, we recommend a specific regular BTM, called \textit{Equal-BTM} and denoted by $ \psi^E $, due to its intuitive fairness property. In an FEE economy, agents may own different sets and amounts of objects. Each step of $ \psi^E $ treats the remaining owners of each object equally by requiring them to supply equal amounts of that object to the trading process in exchange for equal amounts of their respective favorite objects. Thus, if an agent owns more of an object than another, she must utilize the extra amount in the trading process only after the other agent has exhausted her amount of that object. $ \psi^E $ satisfies all fairness criteria mentioned above and meets even stronger criteria in simple economies: \textit{ordinal fairness} and \textit{generalized EENE} (Proposition \ref{prop:ordinalfair}). Ordinal fairness requires that, for any two agents owning the same object, the surplus of one agent at any object in her assignment is no greater than the surplus of the other at the same object, unless the former agent's surplus exceeds the size of the latter's entire assignment, which can occur only when the former owns more of that object. Generalized EENE requires that for any two agents owning the same object, neither agent envies the other's assignment when the comparison is restricted to assignments no larger than their respective endowments. Ordinal fairness implies generalized EENE.

We also provide a characterization of $ \psi^E $. We introduce a property of mechanisms called \textit{endowment-expansion invariance}, which captures the key feature of $ \psi^E $: in any economy, if an agent owns the largest amount of an object, increasing her endowment of that object does not affect other owners' trading opportunities, and if this object is no better than the objects received by other agents in their assignments, the others' assignments remain unchanged. $ \psi^E $ uniquely satisfies sd-efficiency and endowment-expansion invariance (Proposition \ref{prop:characterization}).

In the housing market model, every BTM coincides with TTC. In the house allocation model, every BTM coincides with a simultaneous eating mechanism of \cite{bogomolnaia2001new} that satisfies equal-division lower bound, and every BTM satisfying a mild fairness condition coincides with the probabilistic serial (PS) mechanism (Proposition \ref{prop:house allocation:housing market}). Hence, $ \psi^E $ coincides with TTC and PS in these models, implying new characterizations of both mechanisms (Corollary \ref{corollary:PS} and Corollary \ref{corollary:TTC}). In fact, bounded invariance and ordinal fairness, which we use to describe the fairness of BTMs, are the axioms originally used by \cite{BH2012} and \cite{hashimoto2014two} in their characterizations of PS.\footnote{PS is characterized by sd-efficiency, envy-freeness, and bounded invariance \citep{BH2012}. When objects are no more than agents, ordinal fairness alone characterizes PS \citep{hashimoto2014two}.} 

After studying the FEE model, we turn to the \textit{priority-based allocation} model, in which agents' rights over objects are determined by priorities rather than endowments. Existing extensions of TTC typically assume strict priorities, yet weak priorities are common in practice. For example, in school choice, schools often rank students coarsely based on criteria such as residence, producing large indifference classes \citep{erdil2008s}. To address this, we allow for weak priorities and propose the \textit{priority trading mechanism} (PTM), which generalizes TTC. In PTM, at each step, agents point to their most preferred objects, and each object points to the agents in its highest priority class. Analogous to $ \psi^E $, PTM treats all agents in an object’s highest priority class equally by requiring them to use equal amounts of the object for trading. The outcome of PTM is sd-efficient and fair (Proposition \ref{prop:priority}).

Finally, we consider the \textit{house allocation with existing tenants} (HET) model, a hybrid of the house allocation model and the housing market model. In this model, existing tenants own private endowments, while other objects are publicly owned by all agents. We treat HET as a special case of the priority-based allocation model: all agents have equal priority for public endowments, while existing tenants hold the highest priority for their private endowments, with other agents sharing the next priority level. Applying PTM to this model, we obtain the \textit{eating-trading mechanism} (ETM). Interestingly, ETM can be described as a mixture of the simultaneous eating algorithm and TTC (Proposition \ref{prop:pse}): agents continuously ``eat'' objects at rates determined by the ``you request my house—I get your rate'' rule; however, whenever a group of existing tenants simultaneously wishes to eat one another's private endowments in a way that forms a cycle, they immediately trade these private endowments, possibly in fractional amounts. ETM differs from existing extensions of PS to HET.

The paper is organized as follows. We discuss related literature in the remainder of this section. Section \ref{section:method} introduces our method for defining trading mechanisms. Our description is general. However, to allow readers to have a concrete model in mind while reading that section, we define the FEE model in Section \ref{section2:FEEmodel} and discuss examples that are later used to illustrate our method. Section \ref{section:FEEmech} presents the class of BTMs for the FEE model. Section \ref{section:priority} develops our mechanism for the priority-based allocation model. Section \ref{section:HET} presents our mechanism for the HET model. Section \ref{section:conclusion} concludes with additional results and directions for future research. The appendix contains all proofs and supplementary materials.

\paragraph{Related literature}
This paper contributes to the market design literature by introducing a new approach to defining trading mechanisms for various environments. Our approach employs linear equations to characterize the outcome of a trading process, circumventing the challenge of specifying how trades occur in complex environments. 
Using equations to define market design mechanisms has been explored in the literature. A well-known example is the pseudo-market mechanism introduced by \cite{hylland1979efficient} for the house allocation model, which uses supply-demand equations to define a competitive equilibrium. Another example is the characterization of stable matchings by \cite{azevedo2016supply} in a continuum two-sided matching model, which uses supply-demand equations to define stability.  Similar equations are also used by \cite{han2021blood} to define allocations of blood donation. The equations in these papers define a fixed point, which directly defines the outcome of economies. Differently, equations in our approach define trading processes in discrete steps, with one equation system for each step. The equations do not describe a fixed point and do not directly define the outcome of a trading mechanism.

Other studies have employed equations to address the challenge of defining mechanisms in complex settings. \cite{KestenUnver2014} propose the fractional deferred acceptance (FDA) mechanism for the school choice model with weak priorities. To solve the problem that the algorithm may not terminate in finite steps, they reduce every infinite convergence cycle in the algorithm to a linear equation system. \cite{leshno2021cutoff} present a continuous-time description of TTC in a continuum school choice model to study the role of strict priorities. It solves the difficulty with the description of the standard TTC due to measure issues. Their equations do not involve parameters and reduce to the standard TTC in finite economies.

This paper also contributes to the literature by proposing new mechanisms in well-studied models. Here, we discuss related mechanisms in the literature. In the HET model, \citet{Yilmaz2010} propose the individually rational probabilistic serial (IR-PS) mechanism, which approximates PS while addressing the IR constraints for existing tenants. If an existing tenant views her private endowment as the least preferred object, IR-PS will disregard her IR constraint and treat her as one without private endowment. IR-PS does not reduce to TTC in the housing market model. Section \ref{section:HET} discusses the differences between IR-PS and our ETM.

\cite{AS2011} propose the \textit{controlled consuming} (CC) mechanism, which is a generalization of IR-PS, for the FEE model.\footnote{To be precise, \cite{AS2011} consider a more general model than FEE, since they allow some shares of objects to be publicly owned by all agents.} CC satisfies a fairness axiom different from EENE. However, the authors admit that EENE is more natural for the model and leave an extension of TTC to the model as an open question. In an unpublished paper, \cite{aziz2015generalizing} propose an extension of TTC to the FEE model. To address the challenge of clearing cycles in complex graphs, the author uses exogenous orders of agents and objects to prioritize cycles. This approach violates basic fairness criteria, such as ETE. Similarly, \cite{altuntacs2022trading} define probabilistic variants of TTC by assuming strict priorities.

In the priority-based allocation model with weak priorities, to our knowledge, the literature has not proposed a sd-efficient mechanism based on ordinal preferences in which priorities playing nontrivial roles. \cite{KestenUnver2014} generalize DA to find the most efficient ex-ante stable allocations.  \cite{han2017exante} combines the features of DA and PS to find the most efficient ex-ante fair allocation, where ex-ante fairness is a criterion defined by the author. Both mechanisms reduce to DA under strict priorities. Thus, they differ from PTM.

\cite{aziz2022vigilant} propose a generalization of PS that can deal with a broad class of feasibility constraints, including IR constraints arising from
endowments and stability constraints arising from priorities. Their mechanism degenerates to existing extensions of PS, such as the CC mechanism in the FEE model. The roles of endowments or priorities in our mechanisms cannot be transformed into the type of constraints they study. Thus, our mechanisms differ from theirs.

\cite{Kesten2009} is the first to show that PS is equivalent to a probabilistic version of TTC. The author defines a probabilistic version of TTC, which follows the traditional approach of selecting and clearing cycles.
Since $ \psi^E $ coincides with PS in the house allocation model, our approach provides a new perspective on this equivalence.\footnote{Notably, \citeauthor{Kesten2009}'s mechanism does not coincide with PS step by step, though they are equivalent in outcomes, while $ \psi^E $ coincides with PS step by step.}

Finally, as corollaries of our results, we provide new characterizations of TTC and PS. These results are obtained by changing the set and amounts of objects owned by agents and using a new property capturing how agents respond to the changes. This differs from existing characterizations of the two mechanisms that are typically proven in fixed-size economies and that often rely on axioms such as individual rationality and strategy-proofness for TTC (e.g.,  \citealp{ma1994strategy,miyagawa2002strategy,fujinaka2018endowments,ekici2023pair}), or fairness axioms for PS (e.g., \citealp{BH2012,hashimoto2014two}).


\section{The fractional endowment exchange model}\label{section2:FEEmodel}

Let $ \mathcal{O} $ denote the grand set of objects and $ \mathcal{I} $ the grand set of agents, both finite. A \textbf{fractional endowment exchange} (FEE) economy is represented by a quadruple $ \G=(I,O,\succ_I, \w) $ where $ I \subseteq \mathcal{I}$, $ O \subseteq \mathcal{O}$, $ \succ_I=(\succ_i)_{i\in I} $ is the preference profile of agents, and $ \w=(\w_{i,o})_{i\in I,o\in O} \in [0,1]^{I\times O} $ is the endowment matrix. For each $ i\in I $, $ \succ_i $ represents $ i $'s strict preference over  $ O $ and $ \w_i=(\w_{i,o})_{o\in O}$ represents $ i $'s endowments, where $ \w_{i,o} $ is the amount of object $ o $ owned by $ i $. Each agent demands the equivalent of one unit of an object, and their endowments do not exceed this demand; that is, for all $ i\in I $, $ \sum_{o\in O} \w_{i,o}\le 1 $.  For each $ o\in O $, $ \w_o=(\w_{i,o})_{i\in I}$ represents the distribution of agents' endowments of $ o $.  Let $ q_o = \sum_{i\in I}\w_{i,o} $, which denotes the quota of object $ o $ in the economy. For all $ o, o'\in O $, we write $ o\succsim_i o' $ if $ o\succ_i o' $ or $ o=o' $.

When objects are indivisible, $ \w_{i,o} $ denotes the probability share of object $ o $ initially owned by agent $ i $, and $ q_o $ is typically an integer. However, since we will study properties of mechanisms under varying endowments, we allow objects to have non-integer (including zero) quantities. 
This accommodates another interpretation where objects are infinitely divisible and might be available in any quantity, such as the allocation of time in time banks. This interpretation has been adopted by \cite{Bogomolnaia2015random} and others.

\paragraph{Allocation} Given an economy, an (random) \textbf{assignment} for agent $ i $ is a vector $ p_i\in [0,1]^{O} $ such that $ \sum_{o\in O}p_{i,o}\le 1 $. The \textbf{size} of $ p_i $ is the total amount of objects it includes, defined as $ \norm{p_{i}}=\sum_{o\in O} p_{i,o}$. The \textbf{support} of $ p_i $ is the set of objects it includes, defined as $ \mathrm{supp}(p_i)=\{o\in O: p_{i,o}>0\}$. Following the convention in the literature (e.g., \citealp{bogomolnaia2001new}), we use first-order stochastic dominance to extend agents' preferences to assignments.
Given $ \succ_i $, $p_i$ \textbf{weakly (first-order) stochastically dominates} $p_i'$, denoted by $ p_i\succsim^{sd}_i p'_i $, if $ \sum_{o^{\prime }\succsim _{i}o}p_{i,o^{\prime }}\geq\sum_{o^{\prime }\succsim _{i}o}p'_{i,o^{\prime }}$ for all $ o\in O $. If the inequality is strict for some $ o $, $p_i$ \textbf{strictly stochastically dominates} $p'_i$, denoted by $ p_i \succ^{sd}_i p'_i $.

An assignment $ p'_i $ is a \textbf{subassignment} of $ p_i $ if, for all $ o\in O $, $ p'_{i,o}\le p_{i,o} $. Let $ 2^{p_i} $ denote the set of all subassignments of $ p_i $. Given $ \succ_i $, for any $ x\in [0, \norm{p_i}] $,  $ p_i[x] $ denotes the subassignment of $ p_i $ such that $ \norm{p_i[x]}=x $ and $ p_i[x] \succsim^{sd}_i p'_i $ for all $ p'_i\in 2^{p_i} $ with $ \norm{p'_i}=x $. In words, $ p_i[x] $ is the best subassignment of $ p_i $ with size $ x $ for $ \succ_i $.

\begin{figure}[!htb]
\centering
\small
\begin{tikzpicture}[x=3cm]
	\draw[black,thick,>=latex,line cap=rect]
	(0,0) -- (3,0);
	\foreach \Xc in {0,.7,1.4,2.1,2.5,3}
	{
		\draw[black,thick] 
		(\Xc,-2pt) -- ++(0,4pt);
	}   
	
	
	
	
	
	\node[below,align=left,anchor=north,inner xsep=0pt] 
	at (2.5,.4) {$ \ldots $};
	
	\draw[decorate,decoration={brace,amplitude=8pt}] (2.5,-0.15) -- (0,-0.15)
	node[anchor=north,midway,below=8pt] {$ \scriptstyle x $};
	
	\node[below,align=left,anchor=north,inner xsep=0pt] 		 at (2.5,-.2) {$ \scriptstyle p_i[x] $};
	
	\draw[decorate,decoration={brace,amplitude=8pt}] (0,.15) -- (.7,.15)
	node[anchor=south,midway,above=5pt] {$ \scriptstyle p_{i,o_1} $};
	
	\draw[decorate,decoration={brace,amplitude=8pt}] (0.7,.15) -- (1.4,.15)
	node[anchor=south,midway,above=5pt] {$ \scriptstyle p_{i,o_2} $};
	
	\draw[decorate,decoration={brace,amplitude=8pt}] (1.4,.15) -- (2.1,.15)
	node[anchor=south,midway,above=5pt] {$ \scriptstyle p_{i,o_3} $};
	
	\node[below,align=left,anchor=north,inner xsep=0pt] 
	at (3,-.2) 
	{$ \scriptstyle p_i $};
\end{tikzpicture}
\end{figure}

A convenient way to understand any assignment $p_i$ for $\succ_i$ is to arrange the objects in $ \mathrm{supp}(p_i)$ from most to least preferred along a line such that $[0, p_{i,o_1}]$ represents the amount of the most preferred object $o_1$, $[p_{i,o_1}, p_{i,o_1} + p_{i,o_2}]$ represents the amount of the next preferred object $o_2$, and so on. Then, for any $x \in [0, \norm{p_i}]$, $p_i[x]$ is the assignment represented by the interval $[0, x]$. For this reason, we also refer to $p_i[x]$ as the ``upper part'' of $p_i$ with size $x$.

An allocation specifies an assignment for each agent, subject to the supply constraints of objects. Since we are considering exchange economies without monetary transfers, we focus on allocations resulting from one-for-one exchanges, where each agent receives an amount of objects equal to that of her initial endowments. 

Formally, an \textbf{allocation} is denoted by a matrix $p=(p_{i,o})_{i\in I, o\in O}\in [0,1]^{I\times O}$ such that, for all $ o\in O$, $\sum_{i\in I}p_{i,o}= q_o$, and, for all $ i\in I$, $\sum_{o\in O}p_{i,o}=\sum_{o\in O}\w_{i,o}$. For each $ i\in I $, $p_{i}=(p_{i,o})_{o\in O}$ denotes the assignment received by $ i $.  An allocation $p$ weakly stochastically dominates another $p^{\prime }$, denoted by $p \succsim^{sd}_I p^{\prime }  $, if $p_{i}\succsim^{sd}_i p_{i}^{\prime }$ for all $i\in I$; $ p $ strictly stochastically dominates $ p' $, denoted by $p \succ^{sd}_I p^{\prime }  $, if there further exists an agent $ j $ such that $p_{j} \succ^{sd}_j p_{j}^{\prime }$. 
An allocation $ p $ is \textbf{sd-efficient} if it is not strictly stochastically dominated by any other allocation.  It is \textbf{individually rational} (IR) if, for all $ i\in I $, $ p_i \succsim^{sd}_i \w_i $.

In an allocation $ p $, an agent $ i $ is said to \textbf{envy} another agent $ j $ if $ p_i\not\succsim^{sd}_i p_j $. An allocation $ p $ satisfies \textbf{equal treatment of equals} (ETE) if, for all $ i,j\in I $ with $ \w_i=\w_j $ and $ \succ_i=\succ_j $, $ p_i=p_j $. An allocation $ p $ satisfies \textbf{equal-endowment no envy} (EENE) if, for all $ i,j \in I$ with $ \w_i=\w_j $, $ p_i\succsim^{sd}_i p_j $ and $ p_j\succsim^{sd}_j p_i $. EENE implies ETE.

\paragraph{Mechanism} A \textbf{mechanism} $ \psi $ is a rule that selects an allocation for each economy. When there is no confusion, we represent an economy by its preference profile $ \succ_I$ and use $ \psi(\succ_I) $ to denote the allocation selected by $ \psi $. We use $ \psi_i(\succ_I) $ to denote the assignment to agent $ i $. A mechanism is said to satisfy an efficiency or fairness axiom if its allocations for all economies satisfy the axiom. 
An agent $ i $ can \textbf{manipulate} a mechanism $\psi $ in an economy $ \succ_I $ by reporting $ \succ'_i\neq\succ_i $ if $\psi_{i}(\succsim_{I}) \not\succsim^{sd}_i \psi _{i}(\succsim^\prime_{i},\succsim_{-i})$. An agent $ i $ can  \textbf{strongly manipulate} $\psi $ in $ \succ_I $ by reporting $ \succ'_i \neq \succ_i $ if $\psi _{i}(\succsim^\prime_{i},\succsim_{-i}) \succ^{sd}_i \psi_{i}(\succsim_{I})$. A mechanism $ \psi $ is \textbf{(weakly) strategy-proof} if it cannot be (strongly) manipulated by any agent in any economy.

For any $ \Gamma = (I,O,\succ_I, \w) $, any $ i\in I $, and any $ o \in O $, define $ U(\succsim_i,o)=\{o'\in O: o'\succsim_i o\} $ as the upper contour set of $ o $ in $ \succ_i $, and $\succ_i|_{U(\succsim_i,o)} $ denotes the restriction of $ \succ_i $ to $U(\succsim_i,o)$. 
A mechanism $\psi$ satisfies \textbf{bounded invariance} if, for any $ \Gamma = (I,O,\succ_I, \w) $, any $ i\in I $, any $ o \in O $, and any $ \succ'_i \neq \succ_i $ such that $ \succ_i|_{U(\succsim_i,o)} = \succ'_i|_{U(\succsim'_i,o)} $, we have, for all $ j \in I $, $\psi_{j,o} (\succ_I) = \psi_{j,o}(\succ_{I \backslash \{i\}},\succ'_i)$. Bounded invariance implies that an agent cannot change the allocation of an object by changing the reported ordering of less preferred objects, but it is independent of strategy-proofness.

\paragraph{Classical models} The \textbf{housing market model} is a special case of the FEE model where $ |I|=|O| $, $ q_o=1 $ for all $ o\in O $, and $ \w $ is a permutation matrix. In the \textbf{house allocation model}, all objects are collectively owned by all agents.  In this paper, we regard it as a special case of the FEE model where agents own equal divisions of all objects; that is, for all $ i\in I $ and all $ o\in O $, $ \w_{i,o}=\frac{q_o}{|I|}$. For this sake, we assume that $ \sum_{o\in O}q_o \le |I| $. The general form of the house allocation model, where objects can be more than agents, is dealt with in Section \ref{section:priority} in which we regard it as a special case of the priority-based allocation model where agents have equal priority for all objects.



\paragraph{Examples} In the housing market model, TTC iteratively generates graphs in which agents point to their preferred objects and objects point to their owners. Since every node points to only one other node, cycles must exist and be disjoint, straightforwardly indicating how trades should occur. However, in the FEE model, Example \ref{example:2} shows that the graphs generated from the cycle-clearing idea may be complex, and Example \ref{example:3} shows that even though the generated graphs are simple, clearing cycles may not find an efficient and fair allocation.

\begin{example}\label{example:2}
Consider five agents $ \{1,2,3,4,5\} $ and five objects $ \{a,b,c,d,e\} $. Agents' endowments and preferences are represented by the following tables.
\begin{table}[!ht]
	\centering
	\scalebox{.95}{
		\begin{subtable}{.25\linewidth}
			\raggedright
			\begin{tabular}[c]{c|ccccc}
				& $a$ & $b$ & $c$ & $d$ & $e$ \\\hline
				$1$ & $1$ &  &  \\
				$2$ &  & $1$ &  \\
				$3$ &  &  & $1$ & &  \\
				$4$ &  &  &  & $1$ &  \\
				$5$ &  &  &  &  & $1$ \\
			\end{tabular}
			\subcaption{Discrete endowments}\label{table:example2:discreteendowment}
		\end{subtable}
		\quad
		\begin{subtable}{.35\linewidth}
			\centering
			\begin{tabular}[c]{c|ccccc}
				& $a$ & $b$ & $c$ & $d$ & $e$ \\\hline
				$1$ & $1/2$ & $1/2$ &  \\
				$2$ & $1/2$ & $1/2$ &  \\
				$3$ &  &  & $1/4$ & $1/2$ & $1/4$ \\
				$4$ &  &  & $1/4$ & $1/2$ & $1/4$ \\
				$5$ &  &  & $1/2$ &  & $1/2$ \\
			\end{tabular}
			\subcaption{Fractional endowments}\label{table:example2:fractionalendowment}
		\end{subtable}
		\quad
		\begin{subtable}{.2\linewidth}
			\centering
			\begin{tabular}[c]{ccccc}%
				$\succsim_{1}$ & $\succsim_{2}$ & $\succsim_{3}$ & $\succsim_{4}$ & $\succsim_{5}$ \\\hline
				$c$ & $d$ & $d$ & $a$ & $c$ \\
				$d$ & $c$ & $c$ & $d$ & $e$ \\
				$a$ & $a$ & $a$ & $c$ & $a$\\
				$b$ & $b$ & $e$ & $e$ & $ b $\\
				$e$ & $e$ & $b$ & $b$ & $ d $			
			\end{tabular}
			\subcaption{Preferences}\label{table:example2:preference}
		\end{subtable}
	}
\end{table}

First, consider the case that agents have discrete endowments as shown above. This is a housing market economy. In TTC, step one would generate the following graph in which $ 1,3,4 $ form a cycle and exchange their endowments. At step two, $ 2 $ and $ 5 $ form cycles with their own endowments. Thus, they receive their own endowments.

\begin{figure}[!ht]
	\centering
	\scalebox{.9}{
		\begin{tikzpicture}[bend angle=15,xscale=.7,yscale=.6]
			\node (i2) at (0,0) [agent] {$2$};
			\node (d) at (2,0) [object] {$d$};
			\node (i3) at (4,0) [agent] {$3$};
			\node (c) at (6,0) [object] {$c$};
			\node (i5)  at (8,0) [agent] {$ 5 $};
			\node (b)   at (0,2) [object] {$ b $};
			\node (i4) at (2,2) [agent] {$4$};
			\node (a) at (4,2) [object] {$a$};
			\node (i1) at (6,2) [agent] {$1$};
			\node (e)  at (8,2)  [object] {$ e $};

			\draw[-latex] (a) to (i1);
			\draw[-latex] (b) to (i2);
			\draw[-latex] (c) to (i3);
			\draw[-latex] (d) to (i4);
			\draw[-latex] (e) to (i5);
			\draw[-latex] (i1) to (c);
			\draw[-latex] (i2) to (d);
			\draw[-latex] (i3) to (d);
			\draw[-latex] (i4) to (a);
			\draw[-latex] (i5) to (c);
		\end{tikzpicture}
	}
\end{figure}

Next, suppose that agents change to have fractional endowments as shown above, but their preferences remain unchanged.  
Following the procedure of TTC, we let agents point to most preferred objects and objects point to all of their owners. Then, we would generate the graph in \autoref{figure:example2}. We label the amount of each $ o $ owned by each owner $ i $ besides the edge $ o\rightarrow i $. This graph is more complex than that in the housing market economy. Since the cycles in the graph are \textbf{not disjoint}, the idea of clearing cycles cannot be directly applied. Note that, to find an IR, sd-efficient, and EENE allocation, cycles cannot be cleared arbitrarily.\footnote{We discuss two seemingly intuitive ideas of solving Example \ref{example:2} via clearing cycles in the Online Appendix. The allocations found by both ideas fail to satisfy IR, sd-efficiency, and EENE simultaneously.}

\begin{figure}[!htb]
	\centering
	\scalebox{.9}{
		\begin{subtable}{.5\linewidth}
			\raggedleft
			\begin{tikzpicture}[bend angle=20,xscale=0.9,yscale=0.9]
				\node (o1) at (2,0) [object] {$a$};
				\node (o5) at (8,4.5) [object] {$e$};
				\node (o3) at (5,2) [object] {$c$};
				\node (o4) at (2,4.5) [object] {$d$};
				\node (o2) at (0,0) [object] {$b$};
				\node (i1) at (5,0) [agent] {$1$};
				\node (i2) at (0,4.5) [agent] {$2$};
				\node (i3) at (5,4.5) [agent] {$3$};
				\node (i4) at (2,2) [agent] {$4$};
				\node (i5) at (8,2) [agent] {$5$};
				
				\node at (0.7,2) {\footnotesize $1/2$};
				\node at (2.3,3.5) {\footnotesize $1/2$};
				\node at (3.5,3.8) {\footnotesize $1/2$};
				\node at (3.5,1.7) {\footnotesize $1/4$};
				\node at (5.3,3) {\footnotesize $1/4$};
				\node at (6.6,3.6) {\footnotesize $1/4$};
				\node at (6.5,4.8) {\footnotesize $1/4$};
				\node at (6.5,1.4) {\footnotesize $1/2$};
				\node at (7.6,3.5) {\footnotesize $1/2$};
				\node at (-.3,2) {\footnotesize $1/2$};
				\node at (2.4,-.8) {\footnotesize $1/2$};
				\node at (3.5,.2) {\footnotesize $1/2$};
				
				\draw[-latex] (i1) to (o3);
				\draw[-latex] (i2) to (o4);
				\draw[-latex] (i3) [bend right] to (o4);
				\draw[-latex] (i5) [bend right] to (o3);
				\draw[-latex] (i4) to (o1);
				\draw[-latex] (o1) to (i1);
				\draw[-latex] (o1) to (i2);
				\draw[-latex] (o3) to (i3);
				\draw[-latex] (o3) to (i4);
				\draw[-latex] (o3) [bend right] to (i5);
				\draw[-latex] (o4) [bend right] to (i3);
				\draw[-latex] (o4) to (i4);
				\draw[-latex] (o2) to (i2);
				\draw[-latex] (o2) [bend right] to (i1);
				\draw[-latex] (o5) [] to (i3);
				\draw[-latex] (o5) to (i5);
				\draw[-latex] (o5) to (i4);
			\end{tikzpicture}
		\end{subtable}
		\quad
		\begin{subtable}{.4\linewidth}
			\raggedright
			\begin{tabular}{ll}
				cycle1:& $ 3 \rightarrow d \rightarrow 3 $;\\ cycle2: & $ 5 \rightarrow c \rightarrow 5 $;\\ cycle3: & $ 1 \rightarrow c \rightarrow 4 \rightarrow a \rightarrow 1 $; \\
				cycle4: & $ 2 \rightarrow d \rightarrow 4 \rightarrow a \rightarrow 2 $; \\
				cycle5: & $ 1 \rightarrow c \rightarrow 3 \rightarrow d \rightarrow 4\rightarrow a \rightarrow 1 $.
			\end{tabular}
		\end{subtable}
	}
	\caption{Generated graph  in Example \ref{example:2}}\label{figure:example2}
\end{figure}

\end{example}

\begin{example}\label{example:3}
Consider four agents $ \{1,2,3,4\} $ and three objects $ \{a,b,c\} $. Agents have the endowments and preferences shown in the following tables. 
\begin{figure}[!ht]
	\centering
	\scalebox{.95}{
		\begin{subtable}{.2\linewidth}
			\centering
			\begin{tabular}[c]{c|ccc}
				& $a$ & $b$ & $c$  \\\hline
				$1$ & $1$ &  \\
				$2$ &  & $1$ &  \\
				$3$ &  & $1$ &  \\
				$4$ &  &  & $1$  \\
			\end{tabular}
			\subcaption{Endowments}
		\end{subtable}
		\quad
		\begin{subtable}{.2\linewidth}
			\centering
			\begin{tabular}[c]{cccc}%
				$\succsim_{1}$ & $\succsim_{2}$ & $\succsim_{3}$ & $\succsim_{4}$  \\\hline
				$a$ & $a$ & $c$ & $b$  \\
				$b$ & $c$ & $b$ & $c$  \\
				$c$ & $b$ & $a$ & $a$ 	\\
				
				& & \\
				
			\end{tabular}
			\subcaption{Preferences}
		\end{subtable}
		\quad
		\begin{subfigure}{.5\linewidth}
			\centering
			\begin{tikzpicture}[bend angle=15,xscale=.8,yscale=.8]
				\node (i11) at (0,-1) [agent] {$1$};
				\node (i21) at (2,0) [agent] {$2$};
				\node (i31) at (6,0) [agent] {$3$};
				\node (i41) at (4,0) [agent] {$4$};
				\node (a1) at (0,1) [object] {$a$};
				\node (b1) at (4,2) [object] {$b$};
				\node (c1) at (4,-2) [object] {$c$};
				
				\node at (0.5,0) {\footnotesize $1$};
				\node at (2.8,1.3) {\footnotesize $1$};
				\node at (5.2,1.3) {\footnotesize $1$};
				\node at (4.4,-1) {\footnotesize $1$};

				\draw[-latex] (i11) [bend left] to (a1);
				\draw[-latex] (i21) to (a1);
				\draw[-latex] (i31) to (c1);
				\draw[-latex] (i41) to (b1);
				\draw[-latex] (a1) [bend left] to (i11);
				\draw[-latex] (b1) to (i21);
				\draw[-latex] (b1) to (i31);
				\draw[-latex] (c1) to (i41);
			\end{tikzpicture}
			\subcaption{Generated graph in Example 2}\label{figure:illus:ex2}
		\end{subfigure}
	}
	\caption{Example \ref{example:3}}\label{table:example3}
\end{figure}

There is a unique IR, sd-efficient, and EENE allocation. IR requires that $ 1 $ must receive her endowment $ a $. Among the remaining agents, $ 2 $ and $ 3 $ own equal amounts of $ b $ and both most prefer $ 4 $'s endowment, while $4$ most prefers $b$. EENE requires that $ 2 $ and $ 3 $ must each provide $ 1/2b $ to $ 4 $ in exchange for $ 1/2 c$. So, they each receive the assignment $ (1/2b,1/2c) $.

If we follow the procedure of TTC, we would generate the graph in \autoref{figure:illus:ex2} at step one. Since the two cycles in the graph are disjoint, there is no difficulty to clear both. If clearing them, $ 1 $ would receive $ a $, and $ 3 $ and $4 $ would exchange their endowments. At step two, $ 2 $ would have to retain her endowment $ b $. However, this allocation violates EENE, since $ 2 $ and $ 3 $ own equal endowments, but $ 2 $ envies $ 3 $.
\end{example}

\section{Formal definition of our approach}\label{section:method}

This section outlines our approach to defining a trading process for a given market situation, which lays the foundation for defining trading mechanisms in various models, including those discussed in the paper and others not discussed. Before providing a formal definition, we first use the examples in Section \ref{section2:FEEmodel} to illustrate our approach.

\paragraph{Examples revisited} \autoref{figure:example2} represents a market situation in which agents want to exchange their endowments. To define a trading process for \autoref{figure:example2}, we introduce the following notations. For each agent $i$, let $x_i$ denote the amount of the object she demands that she receives through the trading process. For each object $o$, let $x_o$ denote the total amount of $o$ allocated to agents. It follows immediately that $x_o$ equals the sum of $x_i$ over all agents $i$ who demand $o$. On the other hand, all objects are supplied by their owners. Since we consider balanced exchange, the amount of objects supplied by each agent equals the amount of objects she receives. Our equations involve parameters that determine how the owners of each object divide the right to use the object for trading. Depending on the fairness criterion in the context, the mechanism designer may select different parameters. For \autoref{figure:example2}, let us use parameters requiring that the owners of each object supply equal amounts of the object for trading. Thus, each of 1 and 2 supplies $1/2 x_a$ of $ a$ and $1/2 x_b$ of $ b$. Similarly, each of 3, 4, and 5 supplies $1/3 x_c $ of $ c$ and $1/3 x_e$ of $ e$. Each of 3 and 4 additionally supplies $1/2 x_d $ of $ d$.  Thus, for each $i \in \{1,2\}$, we obtain the equation $x_i = 1/2 x_a + 1/2 x_b$. The equations for other agents are similar and presented below. Since agents cannot supply more objects than their endowments, we use the following inequalities to represent such constraints:

\begin{equation}\label{equation:ex1}
\scalebox{.9}{$\begin{cases}
		x_{a}=x_4,\\
		x_{b}=0,\\
		x_{c}=x_1+x_5,\\
		x_{d}=x_2+x_3,\\
		x_{e}=0,\\
		x_1=1/2x_{a} +1/2x_{b},\\
		x_2=1/2x_{a} +1/2x_{b}\\
		x_3=1/3x_{c} +1/2x_{d}+1/3x_{e},\\
		x_4=1/3x_{c} +1/2x_{d}+1/3x_{e},\\
		x_5=1/3x_{c} +1/3x_{e}.
	\end{cases}
	\text{subject to} \quad
	\begin{cases}
		1/2x_{a}\le 1/2,\\
		1/2x_{b}\le 1/2, \\
		1/3x_{c}\le 1/4, \\
		1/2x_{d}\le 1/2, \\
		1/3x_{e}\le 1/4.
	\end{cases}$
}
\end{equation}

There are more than one solution to this equation system subject to the constraints, but there is a unique coordinate-wise maximum solution satisfying the constraints:

\[
\scalebox{.9}{$\mathbf{x}^*=\left(
\begin{array}{cccccccccc}
	x^*_{1}& x^*_{2} & x^*_{3} & x^*_{4} & x^*_{5} & x^*_{a} & x^*_{b} & x^*_{c} & x^*_{d} & x^*_{e}\\
	1/3&1/3&2/3&2/3&1/6&2/3&0&1/2&1&0
\end{array}
\right).$
}
\]

Selecting the maximum solution  is equivalent to letting agents trade the largest possible amounts of objects. This solution means that 1 receives $1/3 c$ and supplies $1/3 a$, 2 receives $1/3 d$ and supplies $1/3 a$, 3 receives $2/3 d$ and supplies $1/6 c$ and $1/2 d$, 4 receives $2/3 a$ and supplies $1/6 c$ and $1/2 d$, and 5 receives $1/6 c$ and supplies $1/6 c$. Objects $b$ and $e$ are not involved in the trading process, since no agents demand them. After this trading process, object $ d $ is exhausted. We then can write similar equations for the trading processes at the remaining steps. The complete procedure for solving Example \ref{example:2} is presented in the Appendix \ref{section:BTM:procedure}. The found allocation is IR, sd-efficient, and satisfies EENE.

In \autoref{figure:illus:ex2}, we write similar equations and constraints.\footnote{That is, we write the equations $ x_a=x_1 +x_2$, $ x_b=x_4 $, $ x_c=x_3 $, $ x_1=x_a $, $ x_2=1/2x_b $, $ x_3=1/2x_b $, and $ x_4=x_c $, and the constraints $ x_a\le 1 $, $ 1/2x_b\le 1 $, and $ x_c\le 1 $.} In the maximum solution, $ x^*_1=x^*_a=1 $, and $ x^*_\ell=0 $ for all other agents or objects $ \ell $. This means that $ 1 $ receives $ a $ and other agents receive nothing. After $ 1 $ is removed with $ a $, at the next step, the maximum solution to another equation system requires that $ 2 $ and $ 3 $ each receive $ 1/2c $ and $ 4 $ receive $ b $. At the last step, $ 2 $ and $ 3 $ each retain $ 1/2b $. This leads to the unique IR, sd-efficient, and EENE allocation for Example \ref{example:3}.

In the housing market model, our approach is equivalent to clearing cycles: when each agent points to one object and each object points to one agent, for any agent or object $ \ell $, $ x^*_\ell=1 $ if and only if $ \ell $ is involved in a cycle.

\paragraph{Formal definition} 
To facilitate comparison with the traditional cycle-clearing approach, we represent a market situation using a directed graph. 
Consider a directed graph $\mathcal{G}= (\mathcal{V},\mathcal{E},\mathbf{q}) $ where $ \cv $ is a finite set of nodes, $ \ce\subseteq \cv \times \cv$ is a set of directed edges, and $ \mathbf{q}=(q_v)_{v\in \mathcal{V}} $ specifies the quotas for nodes. If $(u,v) \in \mathcal{E}$, it means that node $u$ points to node $v$. Depending on the context, $u$ might be demanding something from $v$, $u$ could be a resource owned by $v$, or other interpretations may apply. Each node points to at least one other node, and no node points to itself.\footnote{In some papers, only agents are represented by nodes, and an agent points to herself when demanding her own endowment. In this paper, agents and objects are represented by distinct nodes.} For every $v \in \mathcal{V}$, let $\mathcal{V}_{\rightarrow v} = \{u \in \mathcal{V} : (u, v) \in \mathcal{E}\}$ represent the set of nodes pointing to $v$, and $\mathcal{V}_{v \rightarrow} = \{u \in \mathcal{V} : (v, u) \in \mathcal{E}\}$ the set of nodes pointed by $v$. 

Nodes engage in trading subject to quotas. To define a trading process, for every $u \in \mathcal{V}$, let $x_u$ denote the total amount of something that $u$ obtains from the nodes in $\mathcal{V}_{u \rightarrow}$. Since $u$ may point to multiple nodes, we introduce parameters $ \{\lambda_{v,u}\}_{v\in \cv_{u\rightarrow}} $ to represent how $u$ divides its demand among those nodes. Specifically, $\lambda_{v,u}$ denotes the proportion of $x_u$ coming from $v$. Hence, $\lambda_{v,u} \in [0,1]$ and $\sum_{v \in \mathcal{V}_{u \rightarrow}} \lambda_{v,u} = 1$. Through choosing these parameters, we can control the trading process to achieve various objectives. For instance, we may set $\lambda_{v,u} = \frac{1}{|\mathcal{V}_{u \rightarrow}|}$, meaning that each node equally divides its demand among the nodes it points to. 
In the trading process, each node's supply equals its demand. For every $v \in \mathcal{V}$, its demand is $x_v$, while its supply is $\sum_{u \in \mathcal{V}_{\rightarrow v}} \lambda_{v,u} x_u$. Therefore, for all $v \in \mathcal{V}$, we require
$$ x_v=\sum_{u\in \mathcal{V}_{\rightarrow v}} \lambda_{v,u} x_{u}. $$
These equations characterize the outcome of the trading process in the graph. We solve the equations subject to the constraints $x_v \leq q_v$ for all $v \in \mathcal{V}$.

For convenience, let $ \bx=(x_v)_{v\in \cv} $ and $ \Lambda=(\lambda_{v,u})_{v,u\in \cv} $, where for all $ v,u\in \cv $, $ \lambda_{v,u}\in [0,1] $, $ \lambda_{v,u}=0 $ if $ v\notin \mathcal{V}_{u\rightarrow} $, and $ \sum_{v\in \cv} \lambda_{v,u}=1$. The equations and constraints can thus be written as $ \Lambda \bx=\bx $ and $ \mathbf{x}\le \mathbf{q} $.
\hide{\begin{equation}\label{equation:abstract}
	\Lambda \bx=\bx,\quad \text{ subject to } \quad \mathbf{x}\le \mathbf{q}.
\end{equation}
}
We say that a solution $ \mathbf{x}^* $ to $ \Lambda \bx=\bx $ is the \textbf{maximum solution} if  $ \bx^*\le \mathbf{q} $, and for every other solution $ \mathbf{x}^\diamond $ satisfying $ \bx^\diamond\le \mathbf{q} $, $ \mathbf{x}^*_v\ge \mathbf{x}^\diamond_v $ for all $ v\in \cv$,  and $ \mathbf{x}^*_v> \mathbf{x}^\diamond_v $ for some $ v\in \cv $.
Lemma \ref{thm:existence} proves the existence of the maximum solution.

\begin{lemma}[\citealp{peterson1982leontief,leonlinear}]\label{thm:existence}
The maximum solution $ \mathbf{x}^* $ to $ \Lambda \bx=\bx $ subject to $ \mathbf{x}\le \mathbf{q} $ exists and is nonnegative.
\end{lemma}

The equation system $\Lambda \mathbf{x} = \mathbf{x}$ describes an ``equilibrium'' similar to the classical Leontief model \citep{leontief1941}. In the Leontief model, finite industries produce goods using inputs from other industries, where the input-output relationships are expressed as linear equations. When the market reaches equilibrium, the coefficient matrix of the equation system is stochastic: each element is in $[0,1]$ and each column sums to one. The coefficient matrix of $\Lambda \mathbf{x} = \mathbf{x}$ is also stochastic. Results on the Leontief model implies the existence of solutions to $\Lambda \mathbf{x} = \mathbf{x}$. We provide a self-contained proof of Lemma \ref{thm:existence} in the Online Appendix.

Trading mechanisms defined in our approach follow the procedure outlined below.

\begin{center}
\textbf{A trading mechanism in our approach}
\end{center}

\begin{itemize}

\item[] \textbf{Step} $ d\ge 1 $:  Given the market situation described by a directed graph $ \mathcal{G}(d)= (\mathcal{V}(d),\mathcal{E}(d),\mathbf{q}(d)) $ and the parameters $ \Lambda(d) $ selected by the mechanism designer, we solve
\begin{equation*}
	\Lambda(d)\mathbf{x}=\mathbf{x}, \quad \text{ subject to } \quad \bx\le \mathbf{q}(d).
\end{equation*}
Let the maximum solution $ \bx^* $ become the outcome of the trading process at this step. Then, for every $ v\in \mathcal{V}(d) $, update its quota: $ q_v(d+1)=q_v(d)-x^*_v $. If a node's quota becomes zero, remove the node. Go to the next step. Stop when all nodes are removed.
\end{itemize}

At each step of the above procedure, at least one node's quota is exhausted. Thus, trading mechanisms defined in our approach must stop in finite steps. Within each step, since linear equations can be solved in polynomial time using well-developed methods and software, computation is not a concern for our approach.


\paragraph{Graphical explanation} While our approach does not identify any graphical component to describe how trades occur, it admits a graphical interpretation that generalizes the well-known description of TTC. We show that any trading process defined in our approach corresponds to clearing \textbf{absorbing sets} in a directed graph.\footnote{Absorbing set is a standard concept in graph theory.  In the housing market model with weak preferences, it has been used by \cite{quint2004houseswapping} to discuss the nonemptiness of the strong core and by \cite{alcalde2011exchange} to propose an extension of TTC.}  Absorbing sets can be viewed as generalizations of cycles. They satisfy conditions provided below, which imply that: (1) in a directed graph, absorbing sets must exist and be disjoint; (2) the nodes in each absorbing set must exchange resources exclusively among themselves, while those not belonging to any absorbing set do not participate in any trading.\footnote{To see that these facts are implied by Definition \ref{definition:absorbing set}, note first that the entire graph has no outgoing edges. If a subset of nodes has no outgoing edges but is not inside connected, it must contain a strict subset with no outgoing edges. Finally, every singleton set is inside connected. Thus, absorbing sets must exist. If two absorbing sets share a node, we can find a directed path between any two nodes in different sets, violating the ``no outgoing edge'' condition. To see that the nodes within each absorbing set must trade among themselves, note that those nodes must obtain nothing from the outside nodes, since there is no outgoing edge. Then, outside nodes must also obtain nothing from those nodes, because otherwise balanced exchange is violated.}


\begin{definition}\label{definition:absorbing set}
In a directed graph $ (\mathcal{V},\mathcal{E},\mathbf{q}) $, a subset $ V\subseteq \cv $ is an \textbf{absorbing set} if
\begin{enumerate}
	
	\item (No outgoing edge) There is no directed path from any node $v\in V $ to any node $ u\notin V $;\footnote{A \textit{directed path} from a node $ v $ to another node $ u $ is a sequence of nodes $ v_1,v_2,\ldots,v_z $ such that $ v_1=v $, $ v_z=u $, and for every $ \ell\in \{1,\ldots,z-1\} $, $ v_\ell $ points to $ v_{\ell+1} $.}
	
	\item (Inside connected) Within $ V $, a directed path from every node to every other node exists.
\end{enumerate}
\end{definition}

In \autoref{figure:example2}, there is a unique absorbing set, $ \{1,2,3,4,5,a,c,d\} $. These are nodes involved in cycles in the graph.  In \autoref{figure:illus:ex2}, there is a unique absorbing set, $ \{1,a\} $. These two nodes form a cycle in the graph. However, the other cycle, $\{3, 4, b, c\}$, is not an absorbing set.

Unlike disjoint cycles that directly indicate how trades should occur, absorbing sets can be complex and do not indicate how the nodes within them should trade with each other. To find the answer, we need to solve $\Lambda \mathbf{x} = \mathbf{x}$. Actually, the maximum solution $\mathbf{x}^*$ to $\Lambda \mathbf{x} = \mathbf{x}$ is obtained by solving the equation system for each absorbing set individually. Specifically, let $V_1, \ldots, V_K$ represent the absorbing sets in a graph, and let $U$ denote the set of nodes outside all absorbing sets. For each absorbing set $V$, let $\Lambda_V$ denote the restriction of $\Lambda$ to the nodes in $V$, and let $ \mathbf{x}^*_{V} $ denote the maximum solution to $ \Lambda_V\mathbf{x}_V=\mathbf{x}_V  $ subject to quota constraints.\footnote{Each $\Lambda_V$ has an eigenvalue of 1, with a positive eigenvector $\tilde{\mathbf{x}}_V$ satisfying $\Lambda_V \mathbf{x}_V = \mathbf{x}_V$. Then, the quota constraints determine the maximum solution $ \mathbf{x}^*_{V} $. See the Online Appendix for details.} Then, the maximum solution $ \mathbf{x}^* $ to $ \Lambda\mathbf{x} = \mathbf{x} $ is 
$
\mathbf{x}^*=(\mathbf{x}^*_{V_1},\ldots,\mathbf{x}^*_{V_K},\mathbf{0}_{U}).
$

\section{Trading mechanisms for fractional endowment exchange}\label{section:FEEmech}

In this section, we apply the approach developed in Section \ref{section:method} to define a class of mechanisms, called \textbf{balanced trading mechanisms} (BTMs), for the FEE model.

Some notations are useful for describing the procedure of BTM. Given an FEE economy $ (I,O,\succ_I, \w) $, at the beginning of each step $ d$, $ I(d) $ denotes the set of remaining agents; $ O(d) $ denotes the set of remaining objects; $ \w(d)=(\w_{i,o}(d))_{i\in I(d),o\in O(d)} $ denotes remaining agents' endowments; $ p(d)=(p_{i,o}(d))_{i\in I,o\in O} $ denotes the found allocation up to the beginning of step $ d $. For every $ i\in I(d) $, $ o_i(d) $ denotes $ i $'s favorite object among $ O(d) $. 

In the trading process at step $ d $, for every $ i\in I(d) $, let $ x_i(d) $ denote the amount of $ o_i(d) $ that $ i $ receives; for every $ o\in O(d) $, let $ x_o(d) $ denote the total amount of $ o $ assigned to agents. On the supply side, the amount $ x_o(d) $ of $ o $ are supplied by its owners. For every $ i\in I(d) $, we use a parameter $ \lambda_{i,o}(d) $ to denote the proportion of $x_o(d) $ supplied by $ i $. Thus, we impose the following conditions on the parameters: for all $ o\in O(d) $ and all $ i\in I(d) $,
\begin{center}
$ \lambda_{i,o}(d)\in [0,1] $, $ \sum_{i\in I(d)}\lambda_{i,o}(d)=1 $, and $ \lambda_{i,o}(d)>0 $ only if $ \w_{i,o}(d)>0 $.
\end{center}  

We impose constraints $ \lambda_{i,o}(d) x_o(d)\le \w_{i,o}(d) $ to ensure that agents do not supply more objects than their endowments.


\begin{center}
\textbf{Balanced Trading Mechanism}
\end{center}

\begin{itemize}
\item[] \textbf{Initialization}: For any $ (I,O,\succ_I, \w) $, let
$ I(1)=I $, $ O(1)=O $, $ \w(1)=\w$, and $ p(1)=\mathbf{0} $.

\item[] \textbf{Step} $ d\ge 1 $: Each $ i\in I(d) $ demands her favorite remaining object $ o_i(d) $. Given parameters $ \Lambda(d)=(\lambda_{i,o}(d))_{i\in I(d),o\in O(d)} $ selected by the mechanism designer, we solve the equation system 
\begin{equation}\label{equation7}
	\begin{cases}
		x_o(d)=\sum_{i\in I(d)} \mathbf{1}\{o_i(d)=o\} x_i(d)\footnotemark  & \text{for all } o\in O(d),\\
		x_i(d)=\sum_{o\in O(d)} \lambda_{i,o}(d) x_o(d) &  \text{for all }i\in I(d),
	\end{cases}
\end{equation}
\footnotetext{By convention, $ \mathbf{1}\{o_i(d)=o\}=1 $ if $ o_i(d)=o $ and otherwise $ \mathbf{1}\{o_i(d)=o\}=0 $.
}
subject to the constraints
\begin{equation}\label{equation9}
	\lambda_{i,o}(d) x_o(d) \le \w_{i,o}(d) \quad \text{ for all }i\in I(d) \text{ and }o\in O(d).
\end{equation}

Let $ \mathbf{x}^*(d)=(x^*_a(d))_{a\in I(d)\cup O(d)} $ denote the maximum solution. Let agents trade endowments according to $ \bx^*(d)$. Then, for all $ i\in I $ and all $ o\in O $, if $ i\in I(d) $ and $ o=o_i(d) $, let $  p_{i,o}(d+1)= p_{i,o}(d)+x^*_i(d)$; otherwise, let $ p_{i,o}(d+1)= p_{i,o}(d) $. For all $ i\in I(d) $ and all $ o\in O(d) $, 
let $ \w_{i,o}(d+1)=\w_{i,o}(d)-\lambda_{i,o}(d)x_o(d) $. 

Let $ I(d+1)=\{i\in I(d):\sum_{o\in O} \w_{i,o}(d+1)>0\} $ and $ O(d+1)=\{o\in O(d):\sum_{i\in I(d+1)} \w_{i,o}(d+1)>0\} $. If $ O(d+1) $ is empty, stop the algorithm and return $p(d+1)$. Otherwise, go to step $ d+1 $.
\end{itemize}

The equation system (\ref{equation7}) and the constraints (\ref{equation9}) are special cases of the general approach in Section \ref{section:method}. Hence, the maximum solution $\mathbf{x}^*(d)$ at each step $d$ exists. In the maximum solution, some constraint $\lambda_{i,o}(d) x_o(d) \leq \w_{i,o}(d)$ must bind, implying that $\w_{i,o}(d+1) = 0$. This means that at least one agent's endowment of an object is exhausted at each step. Therefore, the mechanism must stop in at most $|I| \times |O|$ steps.

BTM is a class of mechanisms because, by varying the equation parameters $ \Lambda(d) $, we obtain different trading processes and allocations for an economy. All of these mechanisms are IR because, at each step, every agent receives an amount of her favorite object by losing an equal amount of endowments. Aggregating all steps, the assignment she receives must weakly stochastically dominate her endowment. All of these mechanisms are sd-efficient because all agents follow their preference orders to demand objects in the algorithm. This ensures that no welfare-improving trade opportunities remain in the found allocation.

\begin{proposition}\label{prop:BTM:efficiency}
Every BTM is individually rational and sd-efficient.
\end{proposition}

However, achieving fairness requires careful selection of the equation parameters. Section \ref{section:BTM:property} is dedicated to discussing the fairness properties of BTMs. Intuitively, $\Lambda(d)$ determines how the owners of each object divide the right to use that object for trading at step $ d $. It is unsurprising that a BTM may fail to find a fair allocation if the parameters are inappropriately chosen. For example, if two agents are identical in terms of endowments and preferences but the parameters favor one agent over the other, then the favored agent will receive a more favorable assignment. Section \ref{section:BTM:property} provides sufficient conditions on $ \Lambda(d) $ to ensure that the outcomes of BTMs satisfy various fairness axioms.

An advantage of our definition of BTMs is its flexibility. It allows mechanism designers to select their preferred $ \Lambda(d) $ based on various factors, such as endowments, preferences, steps, or even the names of agents and objects. This flexibility may be useful in some contexts. However, when aiming for fairness, it might be unintuitive if the parameters depend on factors (e.g., the names of agents) that are irrelevant to the economy's fundamentals. Section \ref{section:BTM:regular} studies a subclass of BTMs in which parameters are selected based solely on the distribution of endowments. Such BTMs are called regular.

Among the class of BTMs, we are particularly interested in one of them, called \textbf{Equal-BTM} and denoted by $ \psi^{E} $, as it embodies an intuitive fairness principle. If mechanism designers in specific applications are uncertain about which BTM to adopt, we recommend using $ \psi^E $. At each step $ d $, $ \psi^E $ uses the parameters:
\[
\lambda_{i,o}(d)=\begin{cases}
\dfrac{1}{|\{j\in I(d):\w_{j,o}(d)>0\}|} & \text{if }\w_{i,o}(d)>0; \\
0 & \text{if }\w_{i,o}(d)=0.
\end{cases}
\]
In words, at each step, the remaining owners of each remaining object use equal amounts of the object for trading, even though they may own different amounts of the object. This implies that if two agents initially own equal amounts of an object $ o $, they will use up their owned amounts of $o$ at the same step; if an agent $ i $ owns more of $ o $ than the other $ j $, they will be treated equally until the step in which $j$ exhausts her owned amount of $ o $, and only after that step $i$ can use her additional amount of $o$ for trading. 
$ \psi^E $ degenerates to TTC in the housing market model and to PS in the house allocation model.

$ \psi^E $ is regular and satisfies all fairness conditions outlined in Section \ref{section:BTM:property}. We study its more precise fairness properties in Section \ref{section:fairness of Equal-BTM} and characterize it in Section \ref{section:characterization of Equal-BTM}.

Section \ref{section:BTM:incentive} discusses the incentive properties of BTMs.

\subsection{Fairness of BTM}\label{section:BTM:property}

ETE and EENE are standard fairness axioms, but they have restrictive power in the FEE model, since many agents may have unequal endowments. In scenarios where no agents have equal endowments, these axioms have no restrictions. This motivates us to introduce a new fairness axiom that generalizes the idea of EENE but applies to any two agents.

The basis of our idea is that envy between two agents is acceptable if they own unequal endowments, but any advantage one agent has over the other must be attributable to the difference in their endowments. For instance, if two agents have slightly different endowments, then neither of them should receive a significantly better assignment than the other's.  Our fairness axiom, \textbf{bounded envy}, stipulates that any agent's envy toward another must be bounded by the envied agent's relative advantage in endowments.

Formally, given an economy, in an allocation $ p $, if an agent $ i $ envies another $ j $ (i.e., $ p_i\not\succsim^{sd}_i p_j $), there must exist $ o\in O $ such that $ \sum_{o'\succsim_i o} p_{j,o'}>\sum_{o'\succsim_i o}p_{i,o'} $. We measure $ i $'s envy towards $ j $ by $ \max_{o\in O} \big[\sum_{o'\succsim_i o} p_{j,o'}-\sum_{o'\succsim_i o}p_{i,o'}\big] $, which represents the maximum cumulative advantage of $ j $'s assignment relative to $ i $'s, according to $ i $'s preferences. We measure $ j $'s advantage in endowments relative to $ i $ by $ \sum_{o\in O:\w_{j,o}> \w_{i,o}}  \big(\w_{j,o}-\w_{i,o} \big) $, which can be viewed similarly as the maximum cumulative advantage of $ j $'s endowments. Note that for $ i $ to envy $ j $, it is not necessary for $ j $ to own more of all objects than $ i $; it suffices that $ j $ owns more of certain objects that can be traded for objects preferred by $ i $. Our measure sums up over all such possible objects.

\begin{definition}\label{defn:bounded:envy}
An allocation $ p $ satisfies \textbf{bounded envy} if, for every distinct $ i,j \in I$,
\begin{equation}\label{equation:boundedenvy}
	\max_{o\in O} \big[\sum_{o'\succsim_i o} p_{j,o'}-\sum_{o'\succsim_i o}p_{i,o'}\big]\le \sum_{o\in O:\w_{j,o}> \w_{i,o}}  \big(\w_{j,o}-\w_{i,o} \big).
\end{equation}
\end{definition}

To illustrate this concept, consider an economy where $ i $ owns one unit of $ a $ and another $ j $ owns $ (1/2b, 1/2c) $, and both regard $ a $ as the worst object. In the unique IR allocation, the two agents receive their own endowments, and $ i $ envies $ j $. Notice that $ j $ does not own more of all objects than $ i $, but the objects $ j $ owns more of contribute to both the measure of $ i $'s envy towards $ j $ and the measure of $ j $'s advantage in endowments. Specifically, we have $ \max_{o\in O} \big[\sum_{o'\succsim_i o} p_{j,o'}-\sum_{o'\succsim_i o}p_{i,o'}\big] =  \sum_{o\in O:\w_{j,o}> \w_{i,o}}  \big(\w_{j,o}-\w_{i,o} \big)=1$. This example demonstrates that the inequality (\ref{equation:boundedenvy}) in Definition \ref{defn:bounded:envy} can bind.

Bounded envy implies EENE. For any two agents with equal endowments, the right-hand side of (\ref{equation:boundedenvy}) is zero, and thus the left-hand side must be zero, meaning no envy between them.

Now, we provide sufficient conditions on parameters to ensure that BTM outcomes satisfy various fairness axioms. In these conditions, when we impose $ \lambda_i(d)=\lambda_j(d) $ for two agents $ i $ and $ j $, it means that for all $o \in O(d)$, $\lambda_{i,o}(d) = \lambda_{j,o}(d)$. It implies that, at step $ d $, $ x_i(d)=x_j(d) $, meaning that the two agents obtain equal amounts of their respective favorite objects, and that for all $ o\in O(d) $, $ \w_{i,o}(d)-\w_{i,o}(d+1)=\w_{j,o}(d)-\w_{j,o}(d+1)$, meaning that the two agents lose equal amounts of each object in their endowments. Simply speaking, the two agents are treated equally at step $ d $.

\begin{definition}\label{defn:fairness:parameter}
A BTM satisfies
\begin{itemize}
	\item[(1)]  \textbf{stepwise equal treatment of equals} (stepwise ETE) if, in any economy, at each step $ d $, for every $ i,j\in I(d) $ such that $ \w_i(d)=\w_j(d)$ and  $ o_i(d)=o_j(d) $, $\lambda_i(d)=\lambda_j(d) $.
	
	\item[(2)]  \textbf{stepwise equal-endowment equal treatment} (stepwise EEET) if, in any economy, at each step $ d $, for every $ i,j\in I(d) $ such that $ \w_i(d)=\w_j(d) $, $ \lambda_i(d)=\lambda_j(d) $.
	
	\item[(3)] \textbf{bounded advantage} if, in any economy, at each step $ d $, for every $ i,j\in I(d) $ and every $ o\in O(d) $ such that $ \w_{i,o}(d)\ge \w_{j,o}(d)>0 $, $ \frac{\w_{i,o}(d)}{\w_{j,o}(d)} \ge \frac{\lambda_{i,o}(d)}{\lambda_{j,o}(d)} \ge 1 $.
\end{itemize}
\end{definition}

The first two conditions apply the ETE idea to the procedure of BTMs.  Stepwise ETE requires that if any two agents own equal amounts of all objects at the beginning of step $ d $ and most prefer the same object at the step, they are treated equally at the step. Stepwise EEET imposes this condition on any two agents who own equal amounts of all objects but not necessarily prefer the same object. 
Bounded advantage generalizes stepwise EEET. At any step $ d $, if $ \w_{i,o}(d)\ge \w_{j,o}(d) >0$, then $\frac{\lambda_{i,o}(d)}{\lambda_{j,o}(d)} \ge 1 $ implies that $  \lambda_{i,o}(d)x_o(d)-\lambda_{j,o}(d)x_o(d) \ge 0 $, meaning that $ i $ uses weakly more of $ o $ than $ j $'s for trading at step $ d $, while $ \frac{\w_{i,o}(d)}{\w_{j,o}(d)} \ge \frac{\lambda_{i,o}(d)}{\lambda_{j,o}(d)} $ implies that $ \lambda_{i,o}(d)x_o(d)-\lambda_{j,o}(d)x_o(d) \le  \w_{i,o}(d)-\w_{j,o}(d)$,\footnote{$ \lambda_{i,o}(d)x_o(d)-\lambda_{j,o}(d)x_o(d)= \lambda_{j,o}(d)x_o(d)\big(\frac{\lambda_{i,o}(d)}{\lambda_{j,o}(d)}-1\big) \le \w_{j,o}(d) \big(\frac{\w_{i,o}(d)}{\w_{j,o}(d)}-1\big)= \w_{i,o}(d)-\w_{j,o}(d)$.} meaning that the more amount of $ o $ that $ i $ uses for trading does not exceed the more of $ o $ she owns relative to $ j $.

\begin{proposition}\label{prop:BTA:fairness}
(1) Every BTM satisfying stepwise ETE satisfies ETE.

(2) Every BTM satisfying stepwise EEET satisfies EENE.

(3) Every BTM satisfying bounded advantage satisfies bounded envy.
\end{proposition}


$\psi^{E}$ is an example of BTMs that satisfies all the conditions defined in Definition \ref{defn:fairness:parameter}.

Another example is the \textbf{Proportional-BTM} (denoted by $\psi^{P}$), which uses the parameter $
\lambda_{i,o}(d)=\frac{\w_{i,o}(d)}{\sum_{j\in I(d)}\w_{j,o}(d)}
$. In words, at each step, the owners of each object use amounts in proportion to their endowments of the object for trading.

BTMs degenerate to well-studied mechanisms in classical models. In the housing market model, every BTM coincides with TTC. In the house allocation model, \cite{bogomolnaia2001new} propose the class of simultaneous eating algorithms (SEAs) to find sd-efficient allocations.  In a SEA, agents ``eat'' probability shares of objects continuously in time. An intuitively fair SEA is the PS algorithm, which gives agents equal eating rates.

When the house allocation model is viewed as a special case of the FEE model where agents own equal divisions of all objects, IR becomes equivalent to \textit{equal-division lower bound}, which is a fairness axiom used in the literature requiring that each agent's received assignment weakly stochastically dominates the equal division. We prove that every BTM coincides with a SEA satisfying equal-division lower bound, and every BTM satisfying stepwise EEET coincides with PS.

\begin{proposition}\label{prop:house allocation:housing market}
	(1) In the housing market model, every BTM coincides with TTC.
	
	(2) In the house allocation model, every BTM coincides with a simultaneous eating mechanism satisfying equal-division lower bound, and every BTM satisfying stepwise EEET coincides with the probabilistic serial mechanism.
\end{proposition}

\subsection{BTM with endowment-dependent parameters}\label{section:BTM:regular}

Since we are considering an endowment exchange model, it is intuitive to focus on a subclass of BTMs where the parameters at each step depend only on the distribution of endowments at the step, and the dependence is governed by an exogenous rule.

Given an endowment matrix $ \w=(\w_{i,o})_{i\in I, o\in O} $, for each $ o \in O$, we use a multiset, denoted by $ \langle \w_o \rangle $, to represent the distribution of agents' endowments of $ o $. Let $ \langle w_o \rangle = \{\!\!\{ w_{i,o} : i \in I,\, w_{i,o} > 0 \}\!\!\} $, where an element appears $ k $ times in $ \langle \w_o \rangle $ if and only if $ k $ agents own that amount of $ o $. Importantly, the multiset does not include the identities of the owners, and the order of elements does not matter. For convenience, we order the elements of $ \langle \w_o \rangle $ in non-decreasing order. For instance, if there are four agents who own zero of $ o $, $ 1/2o $, $ 1/3o $, and $ 1/3o $, respectively, then $ \langle \w_o \rangle=\{\!\!\{1/3,1/3,1/2\}\!\!\} $. We use $ \mathcal{D} $ to denote the set of all multisets that represent the distribution of any object in all possible economies. 

\begin{definition}
A BTM $ \psi $ is \textbf{regular} if there exists a function $ f: [0,1]\times \mathcal{D} \rightarrow [0,1] $ such that for every $ D\in \mathcal{D} $, $ f(0,D)=0 $ and $ \sum_{x\in D} f(x,D)=1 $, and for any endowment matrix $ \w(d)=(\w_{i,o}(d))_{i\in I(d), o\in O(d)} $ at the beginning of any step $ d $, $ \psi $ selects the parameters $ \lambda_{i,o}(d) =f(\w_{i,o}(d),\langle \w_o(d) \rangle)$ for all $ i\in I(d) $ and all $ o\in O(d) $.
\end{definition}


For any regular BTM $ \psi $, let $ f^\psi $ denote the function governing its parameters and call it the parameter function of $ \psi $. For instance, $ \psi^E $ is a regular BTM that uses a constant function: for every $ D\in \mathcal{D} $ and every $ x\in D $,  $ f^{\psi^E}(x,D)=\frac{1}{|D|} $. $ \psi^P $ is another regular BTM that uses a proportional function: for every $ D\in \mathcal{D} $ and every $ x\in D $, $ f^{\psi^P}(x,D)=\frac{x}{\sum_{y\in D}y} $.

All regular BTMs satisfy EENE, because they satisfy stepwise EEET. Moreover, all regular BTMs are \textbf{anonymous} and \textbf{neutral}, because their parameters do not depend on the identities of agents and objects. Formally, a mechanism $ \psi $ is \textit{anonymous} if, for any two economies $ \Gamma= (I,O,\succ_I, \w) $ and $ \Gamma'= (I,O,\succ'_I, \w') $, if there exists a bijection $ g: I\rightarrow I $ such that for all $ i\in I $, $ \succ_i=\succ'_{g(i)} $ and $ \w_i=\w'_{g(i)} $, then for all $ i\in I $, we have $ \psi_i(\Gamma)=\psi_{g(i)}(\Gamma') $. It is \textit{neutral} if, for any two economies $ \Gamma= (I,O,\succ_I, \w) $ and $ \Gamma'= (I,O,\succ'_I, \w') $, if there exists a bijection $ g: O\rightarrow O $ such that for all $ i\in I $ and all distinct $ o,o'\in O $, $ o \succ_i o'$ if and only if $g(o)\succ'_{i} g(o') $, and $ \w_{i,o}=\w'_{i,g(o)} $, then for all $ i\in I $ and all $ o\in O $, we have $ \psi_{i,o}(\Gamma)=\psi_{i,g(o)}(\Gamma') $.

\begin{proposition}\label{prop:regular:fairness}
(1) Every regular BTM is anonymous, neutral, and satisfies EENE. 

(2) A regular BTM $ \psi $ satisfies bounded envy if its parameter function $ f^\psi $ satisfies that, for all $ D\in \mathcal{D} $ and all $ x,y\in D $ such that $ x\ge y $,  $ \frac{x}{y}\ge \frac{f^\psi(x,D)}{f^\psi(y,D)}\ge 1 .$
\end{proposition}

Another appealing property of regular BTM is decomposability: in any economy, splitting any agent $ i $ into multiple sub-agents each owning only one of $ i $'s original endowments does not change $ i $'s assignment, which now equals the combination of the assignments received by these sub-agents, as well as the other agents' assignments. This holds because splitting any agent's endowments does not change the distribution of objects in agents' endowments.

Formally, for any economy $ \Gamma= (I,O,\succ_I, \w) $ and any $ i\in I $, let $ \Gamma_{[i\rightarrow I']}$ denote a new economy $ (I\backslash \{i\} \cup I',O,(\succ_{I\backslash \{i\}},\succ_{I'}), (\w_{I\backslash \{i\}},\w_{I'})) $, where $ I' $ is a set of agents representing the decomposition of $ i $ such that $ |I'|=|\mathrm{supp}(\w_i)| $, $ \sum_{i'\in I'}\w_{i'}=\w_i $, and for all distinct $ i',i''\in I' $, $ \succ_{i'}= \succ_{i''}=\succ_i $ and $ \mathrm{supp}(\w_{i'})\cap \mathrm{supp}(\w_{i''})=\emptyset $.

\begin{definition}
A mechanism $ \psi $ satisfies \textbf{decomposability} if for any FEE economy $ \Gamma= (I,O,\succ_I, \w) $ and any $ i\in I $, $ \psi_i(\Gamma)=\sum_{i'\in I'}\psi_{i'}(\Gamma_{[i\rightarrow I']}) $, and for all $ j\neq i$, $ \psi_j(\Gamma)=\psi_{j}(\Gamma_{[i\rightarrow I']}) $.
\end{definition} 

\begin{lemma}\label{lemma:regular:decomposability}
Every regular BTM satisfies decomposability.
\end{lemma}

Given an economy, if we split every agent into multiple sub-agents each owning only one of the agent's original endowments, we obtain a new economy in which every two agents either own distinct objects, or own the same object but possibly in different amounts. To pin down a regular BTM, it is sufficient to specify its allocations for such simplified economies.

\begin{definition}
An FEE economy $ (I,O,\succ_I, \w) $ is called \textbf{simple} if, for every $ i\in I $, there exists at most one $ o\in O $ such that $ \w_{i,o}>0 $.
\end{definition}

\subsection{Fairness of Equal-BTM in simple FEE economies}\label{section:fairness of Equal-BTM}

Among the class of BTMs, we are particularly interested in $ \psi^E $ because of its intuitively fair procedure. $ \psi^E $ is a regular BTM. To more precisely demonstrate its fairness properties, we analyze it in simple FEE economies. The simplified endowment structure in these economies allows us to study the impact of endowments in a more transparent way.

We first define a generalization of EENE for agents who own the same object but possibly in different amounts. Consider $ i $ and $ j $ who own $ o $ in a simple economy, with $ \w_{i,o}<\w_{j,o} $. In any allocation, $ j $ must receive more of objects than $ i $ does. If $ i $ envies $ j $, it could be because $ j $ receives better objects than $ i $ does, or simply because $ j $ receives more of objects but without receiving better objects. The axiom requires that $ i $'s envy towards $ j $ be eliminated when $ i $ is restricted to receiving no more objects than her endowment.

\begin{definition}
In a simple FEE economy, an allocation $ p $ satisfies \textbf{generalized EENE} if, for every object $ o $ and its every two owners $ i $ and $ j $ with $ 0<\w_{i,o}\le \w_{j,o} $, $ p_j \succsim^{sd}_j p_i $ and $ p_i \succsim^{sd}_i p'_j $ for all  $ p'_j \in 2^{p_j} $ with $ \norm{p'_j}=\norm{p_i} $.
\end{definition}

Generalized EENE is stronger than bounded envy for the owners of the same object. When $ 0<\w_{i,o}< \w_{j,o} $, bounded envy requires that $ i $'s potential envy towards $ j $ be bounded by $\w_{j,o}-\w_{i,o} $, but it does not preclude the possibility that $ i $ prefers a subassignment of $ j $'s assignment with the same size as her endowment.

We prove that $ \psi^E $ satisfies a stronger axiom called \textbf{ordinal fairness}. It is originally used by \cite{hashimoto2014two} to characterize PS in the house allocation model.  
Specifically, in any allocation, an agent's surplus at an object is defined as the amount of weakly better objects she receives. Our ordinal fairness requires that, for any $ i $ and $ j $ with $0< \w_{i,o}\le \w_{j,o} $, $ i $'s surplus at any object in her assignment is no more than $ j $'s surplus at the same object; conversely, $ j $'s surplus at any object in her assignment is no more than $ i $'s surplus at the same object, unless $ i $'s surplus has reached the size of her assignment and $ j $'s surplus exceeds this size.

\begin{definition}
In a simple FEE economy, an allocation $ p $ satisfies \textbf{ordinal fairness} if, for every object $ o $ and its every two owners $ i $ and $ j $ with $ 0<\w_{i,o}\le \w_{j,o} $,

(1) $ p_{i,a}>0\implies\sum_{o'\succsim_i a}p_{i,o'}\le  \sum_{o'\succsim_j a}p_{j,o'} $, and

(2) $ p_{j,a}>0\implies\sum_{o'\succsim_j a}p_{j,o'}\le  \sum_{o'\succsim_i a}p_{i,o'} $, unless $ \sum_{o'\succsim_i a}p_{i,o'}= \norm{p_i}<\sum_{o'\succsim_j a}p_{j,o'}$.
\end{definition}


Ordinal fairness can be intuitively understood through the line interpretations of assignments introduced in Section \ref{section2:FEEmodel}. Recall that any assignment $ p_i $ for agent $ i $ with preference $ \succ_i $ can be arranged along a line within the interval $[0, \norm{p_i}]$.
Now, for agent $ i $, we define a function $ t_i:[0,\norm{p_i}] \rightarrow \mathrm{supp}(p_i) $ such that, for every $ x\in [0,\norm{p_i}] $, $ t_i(x) $ is the worst object in $ \mathrm{supp} (p_i[x])$. Intuitively, $ t_i(x) $ corresponds to the object located at the point $ x $ in the line. We then show that ordinal fairness is equivalent to the condition in Lemma \ref{lemma:ordinalfairness}.

\begin{figure}[!htb]
\centering
\small
\begin{tikzpicture}[x=2.5cm]
	\draw[black,thick,>=latex,line cap=rect]
	(0,0) -- (3,0);
	\foreach \Xc in {0,3}
	{
		\draw[black,thick] 
		(\Xc,-2pt) -- ++(0,4pt);
	} 
	
	
	\node[below,align=left,anchor=north,inner xsep=0pt] 
	at (3.2,.2) 
	{$ p_j $};

	\draw[black,thick,>=latex,line cap=rect]
	(0,0.4) -- (2.5,.4);
	\foreach \Xc in {0,2.5}
	{
		\draw[black,thick] 
		(\Xc,0.33) -- ++(0,4pt);
	}
	
	
	\node[above,align=left,anchor=north,inner xsep=0pt] 
	at (2.7,.7) 
	{$ p_i $};
	
	\draw[black,thick, dotted,>=latex,line cap=rect]
	(1,0) -- (1,0.5);
	
	\draw[black,thick, dotted,>=latex,line cap=rect]
	(2.5,0) -- (2.5,0.38);
	
	\draw[decorate,decoration={brace,amplitude=8pt}] (0,.5) -- (1,.5)
	node[anchor=south,midway,above=5pt] {$ \scriptstyle x  $};

	
	\node[above,align=left,anchor=north,inner xsep=0pt] 
	at (1.15,0.85) 
	{$ \scriptstyle t_i(x) $};
	
	\node[above,align=left,anchor=north,inner xsep=0pt] 
	at (1.15,.05) 
	{$ \scriptstyle t_j(x) $};
	
	\node[above,align=left,anchor=north,inner xsep=0pt] 
	at (2.75,.05) 
	{$ \scriptstyle t_j(y) $};
	
	\draw[fill=black] (1,0) circle [radius=2pt];
	
	\draw[fill=black] (1,.4) circle [radius=2pt];
	
	\draw[fill=black] (2.75,0) circle [radius=2pt];
	
\end{tikzpicture}
\end{figure}

\begin{lemma}\label{lemma:ordinalfairness}
In a simple FEE economy, an allocation $ p $ satisfies ordinal fairness if and only if, for every object $ o $ and its every two owners $ i $ and $ j $ with $0< \w_{i,o}\le \w_{j,o} $,

(1) for all $  x \in [0, \norm{p_i}] $, $ t_i(x)\succsim_i t_j(x) $ and $ t_j(x)\succsim_j t_i(x) $, and

(2) for all $ y\in (\norm{p_i},\norm{p_j}] $, $ t_i(\norm{p_i})\succsim_i t_j(y) $.
\end{lemma}

With Lemma \ref{lemma:ordinalfairness}, it is easy to prove that ordinal fairness implies generalized EENE.\footnote{To see that ordinal fairness is stronger than generalized EENE, consider an economy in which two agents $ i $ and $ j $ own equal endowments $ (1/2a,1/2b) $ but have opposite preferences: $ a\succ_i b $ and $ b\succ_j a $. Then, assigning them the equal assignments $ (1/2a,1/2b) $ satisfies generalized EENE but violates ordinal fairness.}

\begin{lemma}\label{lemma:ordinalfairness:implies:GEENE}
In a simple FEE economy, if an allocation satisfies ordinal fairness, it satisfies generalized EENE.
\end{lemma}

The procedure of $ \psi^E $ ensures that its outcome satisfies the condition in Lemma \ref{lemma:ordinalfairness}. Thus, $ \psi^E $ satisfies ordinal fairness.


\begin{proposition}\label{prop:ordinalfair}
In simple FEE economies, $ \psi^E $ satisfies ordinal fairness.
\end{proposition}

While ordinal fairness imposes strong restrictions on the owners of the same object, the following example shows that it is insufficient to characterize the outcome of $ \psi^E $.

\begin{example}\label{example:fairness:complicated}
Consider a simple economy of four agents $ \{1,2,3,4\} $ and three objects $ \{a,b,c\} $. Agents have the endowments and preferences in the following tables. \autoref{figure:fairness:complicated} shows the graph where agents point to their favorite objects and objects point to their owners.

\begin{table}[!ht]
	\centering
	\scalebox{.95}{
		\begin{subtable}{.2\linewidth}
			\centering
			\begin{tabular}[c]{c|ccc}
				& $a$ & $b$ & $c$  \\\hline
				$1$ & $1$ &  \\
				$2$ &  & $1$ &  \\
				$3$ &  & $1$ &  \\
				$4$ &  &  & $1$  \\
			\end{tabular}
			\subcaption{Endowments}
		\end{subtable}
		\quad
		\begin{subtable}{.2\linewidth}
			\centering
			\begin{tabular}[c]{cccc}%
				$\succsim_{1}$ & $\succsim_{2}$ & $\succsim_{3}$ & $\succsim_{4}$  \\\hline
				$c$ & $a$ & $c$ & $b$  \\
				$a$ & $b$ & $b$ & $c$  \\
				$ b $ & $ c $ & $ a $ & $ a $ 	\\
				
				& & \\
				
			\end{tabular}
			\subcaption{Preferences}
		\end{subtable}
		\quad
		\begin{subtable}{.5\linewidth}
			\centering
			\begin{tikzpicture}[bend angle=15,xscale=.8,yscale=.8]
				\node (i11) at (0,-1) [agent] {$1$};
				\node (i21) at (2,0) [agent] {$2$};
				\node (i31) at (6,0) [agent] {$3$};
				\node (i41) at (4,0) [agent] {$4$};
				\node (a1) at (0,1) [object] {$a$};
				\node (b1) at (4,2) [object] {$b$};
				\node (c1) at (4,-2) [object] {$c$};
				
				\node at (0.3,0) {\footnotesize $1$};
				\node at (2.8,1.3) {\footnotesize $1$};
				\node at (5.2,1.3) {\footnotesize $1$};
				\node at (4.3,-1) {\footnotesize $1$};

				\draw[-latex] (i11) to (c1);
				\draw[-latex] (i21) to (a1);
				\draw[-latex] (i31) to (c1);
				\draw[-latex] (i41) to (b1);
				\draw[-latex] (a1)  to (i11);
				\draw[-latex] (b1) to (i21);
				\draw[-latex] (b1) to (i31);
				\draw[-latex] (c1) to (i41);
			\end{tikzpicture}
			\subcaption{Generated graph at step one of $ \psi^E $}\label{figure:fairness:complicated}
		\end{subtable}
	}
	\caption{Example \ref{example:fairness:complicated}}\label{table:fairness:complicated}
\end{table}

There are two trading opportunities represented by the two cycles, $ b \rightarrow 3 \rightarrow c \rightarrow 4 \rightarrow b $ and $ b \rightarrow 2 \rightarrow a \rightarrow 1 \rightarrow c \rightarrow 4 \rightarrow b $. In any IR and sd-efficient allocation, the sum of the amounts of objects traded in the two cycles must be one. Any IR and sd-efficient allocation can be obtained by dividing this sum among the two cycles, and they all satisfy ordinal fairness. However, $ \psi^E $ selects a specific allocation in which equal amounts of objects are traded in the two cycles. As a result, $ 4 $ receives one unit of $ b $, and all other agents receive $ 1/2 $ of their favorite objects and $ 1/2$ of their own endowments.
\end{example}

Example \ref{example:fairness:complicated} shows that by ensuring fairness among the owners of the same object,  $ \psi^E $ may also achieve certain fairness among the owners of different objects. In this example, because all agents are located in the two cycles initiated by the two owners of $ b $, no envy is achieved among all agents. However, in general cases, depending on the endowments and preferences of agents, situations may become complex, and formalizing the fairness property of $ \psi^E $ for the owners of different objects may become extremely difficult. 
For instance, in the housing market model where $ \psi^E $ is equivalent to TTC, it is unclear how to define the fairness property of $ \psi^E $. This motivates us to characterize $ \psi^E $ using a different approach, which is presented in the next subsection.

\subsection{Characterization of Equal-BTM}\label{section:characterization of Equal-BTM}

Since it is challenging to characterize $ \psi^E $ in a fixed-size economy through its efficiency and fairness properties, we take an approach that varies agents' endowment amounts. An important feature of $ \psi^E $, which contributes to its fairness properties, is that when two agents own different amounts of an object, only after the agent owning less exhausts her endowment of the object, can the agent owning more input her additional endowment into the trading process. Hence, in any economy, if we increase the amount of an object owned by an agent who already owns the most among all owners of that object, the agent will not input the increased amount into the trading process before all owners exhaust their original endowments of the object. If the object is no better than the objects received by all agents in the original economy, the procedure of  $ \psi^E $ in the new economy will remain unchanged until all agents exhaust their original endowments. This observation leads to the following concept.

\begin{definition}\label{definition:endowment-expansion}
In the FEE model, a mechanism $ \psi $ satisfies \textbf{endowment-expansion invariance} if, for any two economies $\Gamma=(I,O,\succ_I, \w) $ and $\Gamma'=(I,O,\succ_I, \w') $ such that $ \w'\ge \w $, and, for every $ o\in O $ such that $ \sum_{i\in I}\w'_{i,o}>\sum_{i\in I}\w_{i,o} $, we have

(1) $\w'_{i,o}> \w_{i,o} \implies \w_{i,o}\ge \w_{j,o}$ for all $j\in I\backslash \{i\}$,

(2) $ \psi_{i,o'} (\Gamma)>0 \implies o'\succsim_i o $,

\noindent then, for all  $ i\in I $, if $ \w_i=\w'_i $, $ \psi_i (\Gamma')=\psi_i (\Gamma) $; if $ \w_i\neq\w'_i $, $ \psi_i (\Gamma')[\norm{\psi_i (\Gamma)}]= \psi_i (\Gamma)$.\footnote{$ \psi_i (\Gamma')[\norm{\psi_i (\Gamma)}] $ denotes the best subassignment of $ \psi_i (\Gamma') $ with the size $ \norm{\psi_i (\Gamma)} $.}
\end{definition}
Conditions (1) and (2) mean that we only expand the endowments of those who already own most of the relevant objects, and all agents weakly prefer their received objects over the objects with increased amounts. The property requires that for agents whose endowments do not change, their assignments neither change, while for agents whose endowments expand, their original assignments remain the ``upper part'' of their assignments in the new economy.

\begin{proposition}\label{prop:characterization}
In the FEE model, a mechanism satisfies sd-efficiency and endowment-expansion invariance if and only if it is $ \psi^E $.
\end{proposition}

The two axioms are independent: assigning all agents their endowments satisfies endowment-expansion invariance but violates sd-efficiency; $ \psi^P $ satisfies sd-efficiency but violates endowment-expansion invariance.

We obtain characterizations of TTC and PS as corollaries. In the housing market model, it is standard to focus on deterministic mechanisms in which agents receive deterministic objects. Then, sd-efficiency can replaced by Pareto efficiency, and endowment-expansion invariance can be reformulated as follows: for any housing market $\Gamma=(I,O,\succ_I, \w) $, if $ \Gamma' $ is a new market obtained from $ \Gamma $ by adding a set of agents $ I' $ with their endowments $ O' $ and, for all $ i\in I $ and all $ o\in O' $, $ \psi_{i} (\Gamma) \succ_i o $, then for all  $ i\in I $, $ \psi_i (\Gamma')= \psi_i (\Gamma)$. 

\begin{corollary}\label{corollary:PS}
In the house allocation model, a mechanism satisfies sd-efficiency and endowment-expansion invariance if and only it is PS.
\end{corollary}

\begin{corollary}\label{corollary:TTC}
In the housing market model, a deterministic mechanism satisfies Pareto efficiency and endowment-expansion invariance if and only it is TTC.
\end{corollary}

\subsection{Bounded invariance and asymptotic strategy-proofness}\label{section:BTM:incentive}

Unlike fairness, which restricts the allocation for each preference profile, incentives restrict how allocations vary when agents change their preferences. Several results in the literature have shown the tension between efficiency, fairness, and strategy-proofness for random assignment mechanisms. In the FEE model, this tension is exacerbated by the IR constraint. Specifically, in the house allocation model, every sd-efficient mechanism satisfying ETE is not strategy-proof \citep{bogomolnaia2001new}. In the FEE model, every IR and sd-efficient mechanism is not weakly strategy-proof \citep{AS2011,aziz2018impossibility}. Hence, no BTM is weakly strategy-proof in the FEE model.

However, certain manipulation strategies are ineffective in regular BTMs. Every regular BTM satisfies \textbf{bounded invariance}, meaning that an agent cannot manipulate the allocation of an object by changing the reported ordering of less preferred objects. This holds because, in the procedure of a regular BTM, the allocation found upon any step depends solely on agents' reported preferences before that step.

Moreover, a subclass of regular BTMs is approximately strategy-proof in large economies, and their fairness property plays a key role. In small economies, an agent may have manipulation power when her endowments or preferences are rare. In large economies, however, any individual's manipulation power becomes negligible when there exist many others with identical endowments and preferences. Specifically, if an agent misreports her preferences in a large economy, the change in the distribution of agents' types (defined by their endowments and preferences) will be small. If a regular BTM is insensitive to such a small change, which is formalized as a continuity property, the resulting allocation will change only slightly. Thus, after manipulation, the agent will receive an assignment close to that received by others with identical endowments but different preferences in the original economy. Since regular BTMs satisfy EENE, the benefit from manipulation diminishes as the economy grows. 

Formally, in an economy $ \Gamma=(I,O,\succ_I,\w) $, define $ \Omega_\Gamma=\{\hat{\w}:\exists i\in I, \ \hat{\w}=\w_i\} $ as the set of endowment types. An agent is type-$ (\hat{\w},\succsim) $ if her endowment is $ \hat{\w} $ and her preference is $ \succ $. Let $ \mathcal{P} $ be the set of all strict preferences. For every $ (\hat{\w},\succ)\in \Omega_\Gamma\times \mathcal{P} $, let $ N(\hat{\w},\succ) $ be the number of type-$ (\hat{\w},\succ) $ agents and $
A(\hat{\w},\succ)=\frac{N(\hat{\w},\succ)}{ |I|}
$ the proportion of type-$ (\hat{\w},\succ) $ agents. 

Given a base economy $ \Gamma=(I,O,\succ_I,\w) $, as the economy grows, we assume that the number of agents increases, but the set of object types $ O $ and the set of endowment types $ \Omega_\Gamma $ remain fixed. A sequence of economies $ (\Gamma^{[n]})_{n=1}^\infty $, where $ \Gamma^{[n]}=(O,I^{[n]},\succ_{I^{[n]}},\w^{[n]}) $, is called \textit{regular}  if  $ \Gamma^{[1]}=\Gamma $, for all $ n\ge 2 $, $ |I^{[n]}|=n|I| $, $ \Omega_{\Gamma^{[n]}}=\Omega_{\Gamma} $, and for all $ (\hat{\w},\succ)\in \Omega_{\Gamma}\times \mathcal{P} $,  $\exists A^{[\infty]}(\hat{\w},\succ)\in (0,1) $ such that
$
\lim_{n\rightarrow \infty}A^{[n]}(\hat{\w},\succ)= A^{[\infty]}(\hat{\w},\succ)$. In words, for every $ (\hat{\w},\succ)\in \Omega_\Gamma\times \mathcal{P} $, the proportion of type-$ (\hat{\w},\succ) $ agents converges to a positive fraction $ A^{[\infty]}(\hat{\w},\succ) $.

A mechanism $ \psi $ is \textbf{asymptotically strategy-proof} if, for any regular sequence of economies $ (\Gamma^{[n]})_{n=1}^\infty $ and any $ \varepsilon>0 $, there exists $ n^*\in \mathbb{N} $ such that, for any $ n>n^* $ and any agent $ i $ in $ \Gamma^{[n]} $, the assignment $ i $ receives from truth-telling weakly stochastically dominates the assignment from misreporting any $ \succ'_i \neq \succ_i $, with an error bounded by $ \varepsilon $:
\[
\sum_{o'\succsim_i o}\psi_{i,o}(\succ_{I^{[n]}})\ge\sum_{o'\succsim_i o}\psi_{i,o}(\succ'_i,\succ_{I^{[n]}\backslash \{i\}})-\varepsilon \quad \text{for all }o\in O.
\]

We now define the subclass of regular BTMs that are asymptotically strategy-proof. It restricts how the parameters used by a regular BTM may change when the distribution of endowments in an economy changes. Specifically, it requires that the parameters are continuous in the distribution of endowments. Formally, a sequence $ \{D^n\}_{n=1}^\infty\subseteq \mathcal{D} $ is said to \textit{converge} to any $ D\in \mathcal{D} $ if, for sufficiently large $ n $, $ |D^n|=|D |$, and for every $ k\in \{1,2,\ldots,|D|\} $, the sequence of the $ k $-th element of $ D^n $ converges to the $ k $-th element of $ D $.

\begin{definition}
A regular BTM $ \psi $ is \textbf{continuous} if its parameter function $ f^\psi $ satisfies that, for every sequence $ \{D^n\}_{n=1}^\infty\subseteq \mathcal{D} $ that converges to any $ D \in \mathcal{D}$, and for every sequence $ \{x^n\}_{n=1}^\infty $ with $ x^n\in D^n $ that converges to any $ x\in D $,  $\lim_{n\rightarrow \infty} f^\psi(x^n,D^{[n]})= f^\psi(x,D) $.
\end{definition}

$ \psi^E $ and $ \psi^P $ are examples of continuous regular BTMs.

\begin{proposition}\label{prop:asymptotic:IC}
Every regular BTM satisfies bounded invariance. Every continuous regular BTM is asymptotically strategy-proof.
\end{proposition}

\section{Priority-based allocation model}\label{section:priority}

Using the approach developed in Section \ref{section:method}, this section defines a mechanism for the priority-based allocation model, where agents trade their priorities for different objects. Our mechanism reduces to the TTC of \cite{abdulkasonmez2003} when priorities are strict.

In the model, each $ o\in O $ is available in an integer quantity $ q_o $ and ranks the agents in $ I $ by a weak order $ \succsim_o $. Two agents $ i$ and $j $ are in priority ties for an object $ o $ if $ i\sim_o j $. An economy is represented by $ (I,O,\succ_I, \succsim_O) $ where $ \succsim_O=(\succsim_o)_{o\in O} $.

We propose the \textbf{priority trading mechanism} (PTM). At each step, agents demand their favorite objects, and only those with the highest priority for an object can use the object for trading. If multiple agents are tied for an object, we let them use equal amounts of the object for trading. In this sense, PTM generalizes $ \psi^E $ in the FEE model.

\begin{center}
\textbf{Priority trading mechanism}
\end{center}

\begin{itemize}
\item[] \textbf{Notations}: Let $ I(d) $, $ O(d) $, $ o_i(d) $, $ p(d) $, and $ \mathbf{x}^*(d) $ be defined as in BTM. For each $ o\in O(d) $, let $ I_o(d) $ denote the set of agents who have the highest priority among $ I(d) $.

\item[] \textbf{Initialization}: Given $ (I,O,\succ_I, \succsim_O) $, let
$ I(1)=I $, $ O(1)=O $, and $ p(1)=\mathbf{0} $.

\item[]\textbf{Step} $ d\ge 1 $: Each $ i\in I(d) $ reports her favorite object $ o_i(d) $. Let $\mathbf{x}^*(d) $ be the maximum solution to the equations
\begin{equation*}
	\begin{cases}
		x_o(d)=\sum_{i\in I(d)}\mathbf{1}\{o_i(d)=o\}\cdot x_i(d) & \text{ for all } o\in O(d),\\ 
		x_i(d)=\sum_{o\in O(d)} \mathbf{1}\{i\in I_o(d)\}\cdot\dfrac{x_o(d)}{|I_o(d)|} & \text{ for all }i\in I(d),
	\end{cases}
\end{equation*}
subject to the constraints
\begin{equation*}
	\begin{cases}
		x_o(d)\le q_o-\sum_{k=1}^{d-1} x^*_o(k) & \text{ for all } o\in O(d),\\
		x_i(d)\le 1-\sum_{k=1}^{d-1} x^*_i(k) & \text{ for all }i\in I(d).
	\end{cases}
\end{equation*}
For all $ i\in I $ and all $ o\in O$, if $ i\in I(d) $ and $o=o_i(d) $, let $ p_{i,o}(d+1)= p_{i,o}(d)+x^*_i(d)$; otherwise, let $ p_{i,o}(d+1)=p_{i,o}(d) $. 
Let  $ I(d+1)=\{i\in I(d):\sum_{k=1}^{d} x^*_i(k)<1\} $ and $ O(d+1)=\{o\in O(d):\sum_{k=1}^{d} x^*_o(k)<q_o \} $. If $ O(d+1)$ or $ I(d+1) $ is empty, stop the algorithm. Otherwise, go to step $ d+1 $.
\end{itemize}

Since PTM resembles $ \psi^E $, we refrain from providing an example. PTM is sd-efficient. Its intuitively fair procedure also ensures fairness in the outcome. However, similar to $ \psi^E $, it is challenging to characterize its precise fairness property.\footnote{We will need to define agents' contingent endowments implied by the priority structure and then compare agents' endowments as in the FEE model.} An evident property of PTM is that an agent with weakly higher priorities for all objects than another agent must receive a weakly better assignment. Formally, an allocation $ p $ satisfies \textbf{no envy towards weakly lower priority} if, for all $ i,j\in I $ such that $ i\succsim_o j $ for all $ o\in O$, $ p_i\succsim^{sd}_i p_j $. This property implies \textbf{equal-priority no envy}: for any $ i,j\in I $ who have equal priority for all objects, $ p_i\succsim^{sd}_i p_j $ and $ p_j\succsim^{sd}_j p_i $. Similar to $ \psi^E $, PTM satisfies bounded invariance. PTM is also asymptotically strategy-proof if, as the economy grows, for any agent, there are sufficiently many others with equal priority but different preferences. Since the result is similar to that for $ \psi^E $, we refrain from providing a formal analysis.

\begin{proposition}\label{prop:priority}
In the priority-based allocation model, PTM satisfies sd-efficiency, no envy towards weakly lower priority, and bounded invariance.
\end{proposition}

The house allocation model in the most general form can be viewed as a special case of the priority-based allocation model where all agents have equal priority for all objects. Then, PTM degenerates to PS in the general house allocation model.

\section{House allocation with existing tenants}\label{section:HET}

The HET model can be viewed as a special case of the priority-based allocation model, where each existing tenant has the highest priority for her private endowment, with other agents being tied, while all agents have equal priority for public endowments. Then, we can apply PTM and obtain a new mechanism for the HET model. Interestingly, this mechanism admits a description that combines the features of TTC and simultaneous eating algorithms.

Formally, in the HET model, a subset of agents $ I_E $, called existing tenants, own a subset of objects $ O_E $, called private endowments, with each $i\in I_E $ privately owning an object in $ O_E $, denoted by $ \w(i) $. The owner of each $ o\in O_E $ is denoted by $ i_o $. The agents in $ I\backslash I_E $ are called newcomers, and the objects in $ O\backslash O_E $ are called public endowments.

We introduce the \textbf{eating-trading mechanism} (ETM). At each step, we let agents points to their favorite objects and private endowments point to their owners. We then check for cycles among existing tenants. If there exists a cycle, let existing tenants in the cycle trade their private endowments immediately. The traded amount in each cycle can be fractional, depending on the residual demands of the existing tenants involved. If no cycles exist, agents ``eat'' their favorite objects according to the following rates: all newcomers' eating rates are set to one; every existing tenant's eating rate is equal to one plus the sum of the eating rates of the agents who are eating her private endowment. This eating-rate rule is referred to as ``you request my house - I get your rate''. ETM degenerates to PS in the house allocation model and to TTC in the housing market model.\footnote{ETM is first introduced in the third chapter of \cite{zhang2017essays}. Using the technique of \cite{CheKojima2010} who prove the asymptotic equivalence between PS and the random priority mechanism in large economies, \cite{zhang2017essays} proves that ETM is asymptotically equivalent to the ``you request my house - I get your turn'' mechanism of \cite{AbduSonmez1999}.}

\begin{center}
\textbf{Eating-trading mechanism}
\end{center}

\begin{itemize}
\item[] \textbf{Notations}: $ I(d) $, $ O(d) $, and $ o_i(d) $ are defined as before. At the beginning of step $ d $, let $ r_i(d) $ denote the residual demand of each $ i\in I(d) $. During step $ d $, let $ s_i(d) $ denote the eating rate of each $ i\in I(d) $, if they are eating objects at the step. 


\item[] \textbf{Initialization}: $ I(1)=I $, $ O(1)=O $, and $ r_i(1)=1 $ for all $ i \in I(1) $. 

\item[] \textbf{Step $ d \ge 1 $}: Let each $ i\in I(d) $ point to her favorite object $ o_i(d) $. Let each $ o\in  O(d)\cap O_E $ point to its owner $ i_o $, if $ i_o\in I(d) $; otherwise, it does not point to any agent.

\begin{itemize}
	\item If there exist cycles among existing tenants, let the agents in each cycle trade an amount of their private endowments instantly, where the traded amount equals the minimum residual demand of the agents involved in the cycle.\footnote{Our mechanism ensures that the residual amount of each private endowment is always weakly more than the residual demand of its owner.}

	\item If there do not exist cycles, let remaining agents simultaneously eat their favorite objects with following rates: for all $ i\in I(d)\backslash I_E $, $ s_i(d)=1 $; for all $ j\in I(d)\cap I_E $, $ s_j(d)=s_{w(j)}(d)+1 $, where  $ s_{\w(j)}(d)=\sum_{i\in I(d):o_i(d)=\w(j)} s_i(d)$. Agents stop eating when some agent's demand is satisfied or when some object is exhausted.
\end{itemize} 

After each step, remove the agents whose demands are satisfied and the objects that are exhausted. When all agents are satisfied or all objects are exhausted, stop the algorithm. Otherwise, go to step $ d+1 $.
\end{itemize}

We present an example to illustrate the procedure of ETM.

\begin{example}\label{example:HET}
Consider an HET economy of six agents $ \{1,2,3,4,5,6\}$ and six objects $ \{a,b,c,d,e,f\}$. Agents have the following preferences. $ 1,2,3,4,5 $ are existing tenants, with private endowments underlined in their preference orders. $ 6 $ is a newcomer, and $ f $ is a public endowment.

\begin{table}[!ht]
	\centering
	\scalebox{.95}{
		\begin{subtable}{.4\linewidth}
			\centering
			\begin{tabular}{cccccc}
				$ \succsim_{1} $ & $ \succsim_{2} $ & $ \succsim_{3} $ & $ \succsim_{4} $ & $ \succsim_{5} $ & $ \succsim_{6} $ \\ \hline
				$ b $ & $ c $ & $ a $ & $ b $ & $ a $ & $ c $ \\
				$ c $ & \underline{$b$} & $ e $ & $ f $ & $ f $ & $ d $ \\
				\underline{$a$} & $ \vdots $  & \underline{$c$} & $ e $ & $ d $ & $ e $\\
				$ \vdots $ & $ \vdots $ & $ \vdots $ & \underline{$d$} & \underline{$e$} &  $ \vdots $ \\
				& &\\ 
			\end{tabular}
			\subcaption{Preferences}
		\end{subtable}
		\quad
		\begin{subtable}{.4\linewidth}
			\centering
			\begin{tabular}[c]{c|cccccc}
				& $a$ & $b$ & $c$ & $d$ & $ e $ & $ f $ \\ \hline
				$ 1 $ & $0$ & $1$ & $ 0 $ & $ 0 $ & $ 0 $ & $ 0 $ \\
				$ 2 $ & $0 $ & $ 0 $ & $1$ & $ 0 $ & $ 0 $ & $ 0 $  \\
				$ 3 $ & $1$ & $ 0 $ & $ 0 $ & $0$ & $ 0 $ & $ 0 $    \\
				$ 4 $ & $0$ & $0$ & $0$ & $0$ & $ 1/3 $& $ 2/3 $    \\
				$ 5 $ & $ 0 $ & $ 0 $ & $ 0$ & $ 1/2 $ & $ 1/6 $ & $ 1/3 $ \\
				$ 6 $ & $ 0 $ & $ 0$ & $ 0$ & $ 1/2 $ & $ 1/2 $& $ 0 $
			\end{tabular}
			\subcaption{Outcome of ETM}
		\end{subtable}
	}
\end{table}

ETM proceeds in the following steps to find the allocation shown above.

\textbf{Step 1:}  $ 1 $ and $ 4 $ point to $ b $, $ 2 $ and $ 6 $ point to $ c $, and $ 3 $ and $ 5 $ point to $ a $. Existing tenants $1,2,3$ form a cycle. So, they exchange their private endowments instantly.

\textbf{Step 2:} $ 4 $ and $ 5 $ point to $ f $, and $ 6 $ points to $ d $.	There are no cycles. The eating rates of $ 5 $ and $ 6 $ are one. The eating rate of $4$ is two since her private endowment $ d $ is being eaten with a rate of one. When $f$ is exhausted, $ 4 $ obtains $ 2/3f $, $ 5 $ obtains $ 1/3f $, and $ 6 $ obtains $ 1/3d $.

\textbf{Step 3:} $4$ points to $e$, and $ 5 $ and $ 6 $ point to $ d $. $ 4 $ and $ 5 $ form a cycle. Their residual demands are $ 1/3 $ and $ 2/3 $, respectively. So, $ 1/3 $ of their private endowments is traded in the cycle. After that, $ 4 $ is satisfied and removed.

\textbf{Step 4:} $ 5 $ and $ 6 $ point to $d$. There are no cycles. The eating rate of all agents is one. When $d $ is exhausted, each agent obtains $ 1/6d $.

\textbf{Step 5:} $ 5 $ and $ 6 $ point to $e$. There is a cycle between $ 5 $ and her private endowment. After clearing the cycle, $ 5 $ obtains $ 1/6e $. Then, $ 5 $ is satisfied and removed.

\textbf{Step 6:}  $ 6 $ points to $ e $ and eats it with a rate of one. So, $ 6 $ obtains the residual amount of $e $ and is satisfied.	
\end{example} 

ETM is sd-efficient. It is IR for existing tenants, meaning that no existing tenant receives positive amounts of objects worse than her private endowment. It is because every private endowment is exhausted no earlier than its owner being removed. ETM satisfies \textbf{no envy towards newcomers}: for every agent $ i $ and every newcomer $j $, $ p_j\succsim^{sd}_j p_i $. It implies that newcomers do not envy each other. ETM also satisfies bounded invariance.

\begin{proposition}\label{prop:pse}
ETM is equivalent to the restriction of PTM to the HET model. So, ETM satisfies individual rationality for existing tenants, sd-efficiency, no envy towards newcomers, and bounded invariance.
\end{proposition}

ETM differs from IR-PS of \cite{Yilmaz2010} in that ETM is essentially a trading mechanism. ETM allows existing tenants to use their private endowments for trading, while IR-PS minimizes the role of private endowments except for addressing the IR constraints of existing tenants.\footnote{IR-PS runs PS until reaching a point where the IR constraint of a group of agents is about to be violated. At this point, IR-PS isolates that group and their remaining acceptable objects as a subproblem where PS is applied independently. IR-PS proceeds as follows in Example \ref{example:HET}. Starting from time 0, agents eat their favorite objects with a constant rate of one. At time $ 1/4 $, the problem is broken into two sub-problems, $m_{1}=\{\{1/2a,1/2b,1/2c\},\{1,2\}\}$ and $m_{2}=\{\{d,e, f\},\{3,4,5,6\}\}$. In the sub-problem $m_{1}$, $1$ and $2$ eat $b $ and $c$ respectively with a rate of one. At time $ 1/2 $, $m_{1}$ is further broken into two sub-problems $m_{11}=\{\{1/4b,1/4c\},\{2\}\}$ and $m_{12}=\{\{1/2a\},\{1\}\}$. In each of $ m_2 $, $ m_{11} $ and $ m_{12} $, IR-PS coincides with PS.} 
This difference leads to different properties of the mechanisms. IR-PS satisfies a different fairness axiom called \textit{no justified envy},\footnote{An allocation $ p $ satisfies no justified envy if, for any  $ i $ and $ j $, $ p_i \succsim^{sd}_i p_j$ if $ p_i $ is individually rational for $ j $. No justified envy implies no envy towards newcomers, because newcomers accept any assignment.} but it is vulnerable to manipulation by existing tenants through truncation strategies. This vulnerability arises because IR-PS does not satisfy upper invariance, a weaker property than bounded invariance.\footnote{Upper invariance requires that an agent cannot change her received amount of an object by changing the reported ordering of less preferred objects. \cite{mennle2020partial} prove that, for random assignment mechanisms, strategy-proofness can be decomposed into upper invariance, lower invariance, and swap monotonicity. So, upper invariance is implied by strategy-proofness. }

\section{Concluding remarks}\label{section:conclusion}

We conclude the paper by discussing additional results and directions for future research. 

While this paper focuses on models using a single trading mechanism to find allocations, some papers in the literature explore multi-stage allocation processes, where trading mechanisms are used to resolve inefficiencies in the outcomes of other mechanisms. For example, in the school choice model with weak priorities, \cite{erdil2008s} introduce a trading algorithm to address inefficiencies in the deferred acceptance (DA) mechanism with exogenous tie-breaking. This approach is extended by \cite{KestenUnver2014}, \cite{erdil2017two}, and \cite{erdil2019efficiency}, among others. Our method can be applied to define trading mechanisms in these multi-stage settings. To illustrate this, in an earlier version of this paper, \cite{YuZhang2021} show that our method can address the complexities in defining the trading algorithm in the second stage of the fractional deferred acceptance and trading (FDAT) mechanism proposed by \cite{KestenUnver2014}. Our method can readily solve the challenge in treating equal-priority students equally in the trading algorithm, on the basis of the allocation produced by FDA in the first stage.

This paper assumes strict preferences. However, in some models, agents may view some objects as indifferent. In a companion paper, \cite{YuzhangFTTCweak} extend the mechanisms defined in this paper to accommodate weak preferences. To preserve efficiency and fairness, we leverage the trading process, with Lemma \ref{thm:existence} remaining crucial. Take $ \psi^E $ as an example. At any step, if an object has been exhausted but an agent who receives a positive amount of the object finds indifferent objects among the remaining ones, we let the agent label the received amount of the exhausted object as an endowment available for trading. In the subsequent trading process, if the agent loses an amount of this labeled endowment, she is compensated with an equal amount of indifferent objects. This labeling continues until no indifferent objects remain. Fairness is maintained by selecting intuitively fair equation parameters in the trading process.

In the housing market model, it is well known that TTC finds the unique allocation in the strong core when all preferences are strict. We may ask whether a similar result holds for BTMs in the FEE model.  When agents' preferences over assignments are defined by the stochastic dominance relation, they are incomplete. As a result, there are more than one way to define the core in the FEE model. \cite{YuzhangXXcore} define the weak core as the set of IR allocations where no coalition of agents can reallocate their endowments among themselves to obtain new assignments stochastically dominating their original ones. The strong core is similarly defined, but requiring that only some members of the coalition be strictly better off. It turns out that the strong core may be empty, and while the weak core is nonempty, it is compatible only with mild fairness: there always exists an allocation in the weak core that satisfies ETE, but there exist economies where every allocation in the weak core violates EENE. Consequently, a regular BTM may never find an allocation in the weak core in some economies.

Several questions remain open for future research. First, the precise fairness property that characterizes the outcome of $ \psi^E $ remains unknown. Although the procedure of $ \psi^E $ is intuitively fair, formalizing this as an outcome-based property is challenging. Second, characterizing PTM in the priority-based allocation model is an open question. While we characterize $ \psi^E $ through endowment decomposition and variation, this approach seems inapplicable to PTM. A characterization of PTM would also inform the characterization of ETM in the HET model. Finally, exploring further applications of our method to address other market design challenges presents promising opportunities for future research.

\bibliographystyle{aea}

\setlength{\bibsep}{0pt plus 0.3ex}
\bibliography{reference}

\newpage

\appendix

\section{Proofs}\label{appendix:propproof}

\begin{proof}[\normalfont\textbf{Proof of Proposition \ref{prop:BTM:efficiency}}]
(IR) Let $ p $ be the outcome of any BTM. At each step $ d $, for all $ i\in I(d) $ and all $ o\in O(d) $, $ o_i(d)\succsim_i o $. Suppose that $ p $ is not IR. Then, there exists $ o^*\in O $ and $ i\in I $ such that $ \sum_{o\succsim_i o^*}p_{i,o}<\sum_{o\succsim_i o^*}\w_{i,o}  $. Let $ d $ be the earliest step after which all objects in $ U(\succsim_i,o^*) $ are exhausted; that is, $ U(\succsim_i,o^*)\cap O(d+1)=\emptyset $ and $ U(\succsim_i,o^*)\cap O(d)\neq\emptyset $. Then, $ i $'s favorite objects from step one to step $ d $ must belong to $U(\succsim_i,o^*) $, and $ \sum_{o\succsim_i o^*}p_{i,o} $ is the cumulative amount of objects obtained by $ i $ up to step $ d $. However, because $ \sum_{o\succsim_i o^*}p_{i,o}<\sum_{o\succsim_i o^*}\w_{i,o}  $, there must exist $ o'\succsim_i o^* $ such that $ \w_{i,o'}(d+1)>0 $. This is a contradiction. 

(Sd-efficiency)  Define a binary relation $ \rhd $ on $ O $ such that, for every distinct $ o,o'\in O $, $ o\rhd o' $ if there exists $ i \in I$ such that $ o\succ_i o' $ and $ p_{i,o'}>0 $. \cite{bogomolnaia2001new} have shown that $ p $ is sd-efficient if and only if $ \rhd $ is acyclic. In the procedure of BTM, an agent demands an object only after better objects are exhausted. So, for any $ o\succ_i o' $ and $ p_{i,o'}>0 $, $ o $ must be exhausted earlier than $ o' $. This implies that $ \rhd  $ is acyclic.
\end{proof}

\begin{proof}[\normalfont\textbf{Proof of Proposition \ref{prop:BTA:fairness}}]
(1) For any BTM satisfying stepwise ETE, if any $ i$ and $j $ own equal endowments and have identical preferences, at step one, $ \lambda_{i}(1)=\lambda_{j}(1) $. So, they obtain equal amounts of the same favorite object and lose equal amounts of each endowment at step one. At step two, $ \w_i(2)=\w_j(2) $ and identical preferences imply $ \lambda_{i}(2)=\lambda_{j}(2) $. So, again, they obtain equal amounts of the same favorite object and lose equal amounts of each endowment at step two. This inductively holds for all remaining steps. So, the two agents must receive equal assignments.

(2) For any BTM satisfying stepwise EEET, if any $ i$ and $j $ own equal endowments, then at step one, $ \lambda_{i}(1)=\lambda_{j}(1) $. So, they obtain equal amounts of their respective favorite objects and lose equal amounts of each endowment at step one. At step two, $ \w_i(2)=\w_j(2) $ implies $ \lambda_{i}(2)=\lambda_{j}(2) $. So, again, they obtain equal amounts of their respective favorite objects and lose equal amounts of each endowment at step two. This inductively holds for all remaining steps. So, the two agents weakly prefer their own assignment over the other's.

(3) In any economy, let $ p $ be the outcome of any BTM satisfying bounded advantage. 
Suppose that an agent $ i $ envies another $ j $. Let $ o^* $ be the solution to $ \max_{o\in O} \big[\sum_{o'\succsim_i o} p_{j,o'}-\sum_{o'\succsim_i o}p_{i,o'}\big] $.
Let $ d $ be the step after which all objects $ o $ with $o\succsim_i o^* $ are exhausted; that is, $U(\succsim_i, o^*) \cap O(d+1)=\emptyset $ and $ U(\succsim_i, o^*) \cap O(d) \neq\emptyset $. Then, $ \sum_{o\succsim_i o^*}p_{i,o}=\sum_{o\in O}\big(\w_{i,o}-\w_{i,o}(d+1)\big) $ and $ \sum_{o\succsim_i o^*}p_{j,o}\le\sum_{o\in O}\big(\w_{j,o}-\w_{j,o}(d+1)\big) $. So, 
\[
\sum_{o\succsim_i o^*}p_{j,o}-\sum_{o\succsim_i o^*}p_{i,o}
\le\sum_{o\in O}\big[\big(\w_{j,o}-\w_{j,o}(d+1)\big)-\big(\w_{i,o}-\w_{i,o}(d+1)\big)\big].
\]

For every $ o\in O $ such that $ \w_{i,o}\ge \w_{j,o} $, bounded advantage implies that, for all $ 1\le d'\le d $, $ \w_{i,o}(d'+1)\ge \w_{j,o}(d'+1) $ and $ \lambda_{i,o}(d')\ge \lambda_{j,o}(d') $. So, 
$
\w_{j,o}-\w_{j,o}(d+1)=\sum_{d'=1}^d \lambda_{j,o}(d')x^*_o(d')\le  \sum_{d'=1}^d \lambda_{i,o}(d')x^*_o(d')=\w_{i,o}-\w_{i,o}(d+1).
$
Equivalently,
\[
\big(\w_{j,o}-\w_{j,o}(d+1)\big)-\big(\w_{i,o}-\w_{i,o}(d+1)\big)\le  0.
\] 
For every $ o\in O $ such that $ \w_{i,o}< \w_{j,o} $, bounded advantage implies that, for all $ 1\le d'\le d $, $ \w_{i,o}(d'+1)\le \w_{j,o}(d'+1) $ and $ \lambda_{i,o}(d')\le \lambda_{j,o}(d') $. In particular, $ \w_{i,o}(d+1)\le \w_{j,o}(d+1)  $. So,
\[
\big(\w_{j,o}-\w_{j,o}(d+1)\big)-\big(\w_{i,o}-\w_{i,o}(d+1)\big)\le \w_{j,o}-\w_{i,o}.
\]
Therefore, $
\sum_{o\succsim_i o^*}p_{j,o}-\sum_{o\succsim_i o^*}p_{i,o}
\le \sum_{o\in O:\w_{i,o}< \w_{j,o}} \big(\w_{j,o}-\w_{i,o} \big)$. 
\end{proof}

\begin{proof}[\normalfont\textbf{Proof of Proposition \ref{prop:house allocation:housing market}}]
	(1) In the housing market model, at the first step of TTC, suppose that a group of agents $ \{i_1,i_2,\ldots,i_k\} $ forms a cycle in which $ i_1 $ points to $ i_2 $'s endowment denoted by $ o_2 $, $ i_2 $ points to $ i_3 $'s endowment denoted by $ o_3 $, $ \ldots $, $ i_k $ points to $ i_1 $'s endowment denoted by $ o_1 $. Then, at the first step of any BTM, there exist equations $ x_{i_1}=x_{o_1} $, $ x_{o_1}=x_{i_k} $, $ x_{i_k}=x_{o_k} $, $ \ldots $, $ x_{i_2}=x_{o_2} $ and $ x_{o_2}=x_{i_1} $. Every constraint takes the form $ x_\ell \le 1 $. In the maximum solution to the equation system, it must be that $ x^*_{i_1}=\cdots=x^*_{i_k} =1$, meaning that those agents receive their preferred objects as in TTC. Conversely, in the maximum solution to the equation system at the first step of any BTM, if $ x^*_{i_1}=1 $ for any agent $ i_1 $, then $ x^*_{o_2}=1 $, where $ o_2 $ is the object demanded by $ i_1 $, and $ x^*_{i_2}=1 $, where $ i_2 $ is the owner of $ o_2 $, and so on. Since there are finite agents, we must find a group of agents $ \{i_1,i_2,\ldots,i_k\} $ who demand each other's endowment as forming a cycle. Then, this cycle must exist at the first step of TTC. These arguments inductively hold for all remaining steps.  
	
	(2) In the house allocation model, consider any economy. Suppose that the procedure of a BTM has $ n $ steps. At every step $ d$, every $ i\in I $ receives the amount $ p_{i,o_i(d)}(d+1)-p_{i,o_i(d)}(d) $ of her preferred object $ o_i(d) $.  We then arbitrarily choose $ n $ rational numbers $ t_1,t_2,\ldots,t_n \in [0,1] $ such that $ 0=t_0<t_1<t_2<\cdots<t_n= 1 $. We will describe the procedure of the BTM as a SEA procedure in which each step $ d $ starts from $ t_{d-1} $ to $ t_d $. For every $ i $, arbitrarily choose an eating rate function $ s_i:[0,1]\rightarrow \Re_+ $ such that for each step $ d $ of the BTM, $ \int_{t_{d-1}}^{t_d}s_i(t)=p_{i,o_i(d)}(d+1)-p_{i,o_i(d)}(d) $. By choosing these rate functions, it is clear that at any $ t\in [t_{d-1},t_d] $, $ i $ eats $ o_i(d) $. So, the constructed SEA procedure coincides with the BTM procedure. Since agents own equal divisions of all objects and every BTM is IR, the constructed SEA procedure satisfies equal-division lower bound.
	
	If a BTM satisfies stepwise EEET, at the first step, agents receive equal amounts of their respective favorite objects, and after that, they still own equal endowments. This inductively holds in each remaining step. So, the BTM coincides with PS.
\end{proof}

\begin{proof}[\normalfont \textbf{Proof of Proposition \ref{prop:regular:fairness}}]
(1) For any regular BTM $ \psi $, consider any two economies $ \Gamma= (I,O,\succ_I, \w) $ and $ \Gamma'= (I,O,\succ'_I, \w') $ such that for some bijection $ g: I\rightarrow I $ and for all $ i\in I $, $ \succ_i=\succ'_{g(i)} $ and $ \w_i=\w'_{g(i)} $.  It is easy to verify that, at each step of $ \psi $ in $ \Gamma $, if the role of each $ i\in I $ is replaced by $ g(i) $, we obtain the corresponding step of $ \psi $ in $ \Gamma' $. So, $ \psi $ is anonymous. Similarly, we can prove that $ \psi $ is neutral.

At the beginning of any step $ d $, for every $ i,j\in I(d) $ with $ \w_i(d)=\w_j(d) $ and for every $ o\in O(d) $, $ \lambda_{i,o}(d)=f(\w_{i,o}(d),\langle \w_o(d) \rangle)=f(\w_{j,o}(d),\langle \w_o(d) \rangle)=\lambda_{j,o}(d)$. So, $ \psi $ satisfies Stepwise EEET. By Proposition \ref{prop:BTA:fairness}, $ \psi $ satisfies EENE.

(2) For any regular BTM $ \psi $ with its parameter function satisfying the condition, consider any economy appearing at the beginning of any step $ d $ with an endowment matrix $ \w(d)=(\w_{i,o}(d))_{i\in I(d), o\in O(d)} $. Consider any $ i,j\in I(d) $ and any $ o\in O(d) $ such that $ \w_{i,o}(d)\ge \w_{j,o}(d)>0 $. Then, we have $ \frac{\w_{i,o}(d)}{\w_{j,o}(d)} \ge \frac{f(\w_{i,o}(d),\langle \w_o(d) \rangle)}{f(\w_{j,o}(d),\langle \w_o(d) \rangle)} \ge 1 $. So, $ \psi $ satisfies bounded advantage. By Proposition \ref{prop:BTA:fairness}, it satisfies bounded envy.
\end{proof}

\begin{proof}[\normalfont \textbf{Proof of Lemma \ref{lemma:regular:decomposability}}]

Let $ \psi $ be a regular BTM. For any economy $ \Gamma= (I,O,\succ_I, \w) $ and any $ i\in I $, we compare the procedures of $ \psi $ in $ \Gamma $ and $ \Gamma_{[i\rightarrow I']} $. We want to prove that, at each step, the assignment received by $ i $ in $ \Gamma $ equals the sum of the assignments received by all $ i' \in I $ in $ \Gamma_{[i\rightarrow I']} $, while every other agent receives the same assignment in the two economies. To differentiate between the parameters and solutions in the equation systems for the two economies, we use $ \lambda $ and $ x$ for $ \Gamma $, and use $ \lambda' $ and $ x'$ for $ \Gamma_{[i\rightarrow I']} $. Note that for every $ o\in O $, $ \langle \w_o \rangle $ do not change in the two economies. 

At step one of $ \psi $, for every $ o\in O $, if $ \w_{i,o}>0 $, in $ \Gamma $, $ \lambda_{i,o}(1)=f^\psi(\w_{i,o},\langle \w_o \rangle)$, whereas in $ \Gamma_{[i\rightarrow I']}$, there exists a unique $ i'\in I' $ such that $ \w_{i',o}=\w_{i,o} $, implying that $ \lambda'_{i',o}(1)=f^\psi(\w_{i,o},\langle \w_o \rangle)=\lambda_{i,o}(1) $. For every $ i''\in I'\backslash \{i'\} $, $ \lambda'_{i'',o}(1)=0 $. If $ \w_{i,o}=0 $, in $ \Gamma $, $ \lambda_{i,o}(1)=0 $, and in $ \Gamma_{[i\rightarrow I']}$, for every $ i'\in I' $, $ \lambda'_{i',o}(1)=0 $. The parameters for all $ j\in I\backslash \{i\} $ do not change in the two economies. In the equation systems, in $ \Gamma $, $ x_i(1)=\sum_{o\in O(1)} \lambda_{i,o}(1) x_o(1)$, whereas in $ \Gamma_{[i\rightarrow I']} $, for every $ i'\in I' $, $x'_{i'}(1)= \sum_{o\in O(1)} \lambda'_{i',o}(1) x'_o(1)$. The equations for the other agents do not change in the two economies. For $ i $'s favorite object $ o_i(1) $, in $ \Gamma $, $ x_{o_i(1)}=\sum_{j\in I\backslash
	\{i\}}\mathbf{1}\{o_j(1)=o_i(1)\} x_j(1)+x_i(1)$, whereas in $ \Gamma_{[i\rightarrow I']} $, $ x'_{o_i(1)}=\sum_{j\in I\backslash
	\{i\}}\mathbf{1}\{o_j(1)=o_i(1)\} x'_j(1)+\sum_{i'\in I'} x'_{i'}(1)$. The equations for the other objects do not change in the two economies. The quota constraints in the two economies are the same. Comparing the equation systems in the two economies implies that in the maximum solutions, $ x^*_i(1)= \sum_{i'\in I'} x^{\prime *}_{i'}(1)$ and for every $ j\in I\backslash \{i\} $, $ x^*_j(1)=x^{\prime *}_j(1) $. 
The above argument inductively holds for each remaining step.
\end{proof}

\begin{proof}[\normalfont \textbf{Proof of Lemma \ref{lemma:ordinalfairness}}]
(\textbf{If}) Let $ p $ be an allocation satisfying the stated condition. Consider any object $ o $ and its any two owners $ i $ and $ j $ with $0< \w_{i,o}\le \w_{j,o} $. For any object $ a $ such that $ p_{i,a}>0$, let $ x=\sum_{o'\succsim_i a}p_{i,o'} $. Then, $ t_i(x)=a $, and $ t_j(x) \succsim_j t_i(x)$. This implies that the amount of weakly better objects than $ a $ received by $ j $ is at least $ x $. Therefore, $ \sum_{o'\succsim_i a}p_{i,o'} \le  \sum_{o'\succsim_j a}p_{j,o'}$. Similarly, for any object $ a $ such that $ p_{j,a}>0$, if $ \sum_{o'\succsim_j a}p_{j,o'}\le \norm{p_i} $, we can switch the roles of $ i $ and $ j $ in the above argument to prove that $\sum_{o'\succsim_j a}p_{j,o'}\le  \sum_{o'\succsim_i a}p_{i,o'} $. If $ \sum_{o'\succsim_j a}p_{j,o'}> \norm{p_i} $, then there exists $ y\in (\norm{p_i},\norm{p_j}] $ such that $ t_j(y)=a $. Since $ t_i(\norm{p_i})\succsim_i t_j(y) $, it implies that $ \sum_{o'\succsim_i a}p_{i,o'} =\norm{p_i} $.

(\textbf{Only if}) Let $ p $ be an allocation satisfying ordinal fairness. Consider any object $ o $ and its any two owners $ i $ and $ j $ with $0< \w_{i,o}\le \w_{j,o} $. Then, $ \norm{p_i}=\w_{i,o}\le \w_{j,o}=\norm{p_j} $. If there exists $  x \in [0, \norm{p_i}] $ such that $ t_i(x)\succ_j t_j(x) $, then $ \sum_{o'\succsim_i t_i(x)}p_{i,o'}\ge x > \sum_{o'\succsim_j t_i(x)}p_{j,o'}$, which violates ordinal fairness. Similarly, we can obtain a contradiction if $ t_j(x)\succ_i t_i(x) $ for some $  x \in [0, \norm{p_i}] $. So, for all $  x \in [0, \norm{p_i}] $, $ t_i(x)\succsim_i t_j(x) $ and $ t_j(x)\succsim_j t_i(x) $. If $ \norm{p_i}< \norm{p_j} $, for any $ y\in (\norm{p_i},\norm{p_j}] $, $ \sum_{o'\succsim_j t_j(y)}p_{j,o'}\ge y > \norm{p_i} $. Then, ordinal fairness requires that $ \sum_{o'\succsim_i t_j(y)}p_{i,o'}= \norm{p_i}  $. This implies that $ t_i(\norm{p_i})\succsim_i t_j(y) $.
\end{proof}

\begin{proof}[\normalfont \textbf{Proof of Lemma \ref{lemma:ordinalfairness:implies:GEENE}}]
In a simple FEE economy, let $ p $ be an allocation satisfying ordinal fairness. Consider any object $ o $ and its any two owners $ i $ and $ j $ with $ 0<\w_{i,o}\le \w_{j,o} $. We first prove that, for all  $ p'_j \in 2^{p_j} $ with $ \norm{p'_j} = \norm{p_i} $, $ p_i \succsim^{sd}_i p'_j $. 
If this does not hold for some $ p'_j $ satisfying the condition, then there exists an object $ a $ such that $ \sum_{o' \succsim_i a} p_{i,o'} < \sum_{o' \succsim_i a} p'_{j,o'} \leq \norm{p_i}$. Since $ p $ satisfies ordinal fairness, by Lemma \ref{lemma:ordinalfairness}, for any $x \in  (\sum_{o' \succsim_i a} p_{i,o'}, \sum_{o' \succsim_i a} p'_{j,o'}) $, any $y_1 \in [x, \norm{p_i} ] $, and any $y_2 \in (\norm{p_i}, \norm{p_j}) $, $a \succ_i t_i (x) \succsim_i t_i (y_1) \succsim_i t_j (y_1)$ and $a \succ_i t_i (x) \succsim_i t_i (\norm{p_i}) \succsim_i t_j (y_2)$.
Therefore, for every object $ o' \in \mathrm{supp} (p_j) $ such that $ o'\succsim_i a $, it must be that $ o' \succsim_j t_j (x) $. This implies that $ \sum_{o' \succsim_i a} p'_{j,o'} \le \sum_{o' \succsim_i a} p_{j,o'} \le x < \sum_{o'\succsim_i a} p'_{j,o'}$, which is a contradiction. 
Similarly, we can prove that $ p_j \succsim^{sd}_j p_i $.
\end{proof}

\begin{proof}[\normalfont \textbf{Proof of Proposition \ref{prop:ordinalfair}}]
Let $ p $ denote the outcome of $ \psi^E $ in any simple FEE economy. Consider any object $ o $ and its any two owners $ i $ and $ j $ with $0< \w_{i,o}\le \w_{j,o} $. Then, $ \norm{p_i}=\w_{i,o}\le \w_{j,o}=\norm{p_j} $. For any $  x \in [0, \norm{p_i}] $, there exists a step $ d $ in which the cumulative amount of objects received by the two agents reaches $ x $. Therefore, $ t_i(x) $ and $ t_j(x) $ correspond to the objects received by the two agents at step $ d $. Since the two agents receive their respective favorite objects at that step, $ t_i(x)\succsim_i t_j(x) $ and $ t_j(x)\succsim_j t_i(x) $. If $ \w_{i,o}<\w_{j,o} $, then $ j $ inputs her additional endowment of $ o $ into trading only after $ i $ has exhausted her endowment of $ o $. Therefore, for any $ y\in (\norm{p_i},\norm{p_j}] $, $ t_j(y) $ is an object no better than any object received by $ i $ for $ \succ_i $. This means that $ t_i(\norm{p_i})\succsim_i t_j(y) $. By Lemma \ref{lemma:ordinalfairness}, $ p $ satisfies ordinal fairness. 
\end{proof}

\begin{proof}[\normalfont\textbf{Proof of Proposition \ref{prop:characterization}}]
(\textbf{If}) Sd-efficiency of $ \psi^E $ has been proved by Proposition \ref{prop:BTM:efficiency}. We prove that $ \psi^E $ satisfies endowment-expansion invariance. 

Consider any two economies $\Gamma=(I,O,\succ_I, \w) $ and $\Gamma'=(I,O,\succ_I, \w') $ such that $ \w'\ge \w $ and for every $ o\in O $ such that $ \sum_{i\in I}\w'_{i,o}>\sum_{i\in I}\w_{i,o} $,

(1) $\w'_{i,o}> \w_{i,o} \implies \w_{i,o}\ge \w_{j,o}$ for all $j\in I\backslash \{i\}$,

(2) $ \psi^E_{i,o'} (\Gamma)>0 \implies o'\succsim_i o $,

For each $ o\in O $, let $ I_o $ denote the set of owners of $ o $ in $ \Gamma $. 

First, for all $ o\in O $ and all $ i\in I_o $ such that $ \w'_{i,o}> \w_{i,o}$, since $ \w_{i,o}\ge \w_{j,o} $ for all $ j\in I\backslash \{i\} $, the procedure of $ \psi^E $ ensures that $ i $'s increased endowment of $ o $ will not be traded before all $ j\in I_o $ have exhausted their original endowment of $ o $. Second, since for all $ o\in O $ such that $ \sum_{i\in I}\w'_{i,o}>\sum_{i\in I}\w_{i,o} $, $ \psi^E_{i,o'} (\Gamma)>0 \implies o'\succsim_i o $, every agent will not demand the increased amount of every such $ o $ before the original amounts of the objects the agent receives in $ \psi^E(\Gamma) $ are exhausted. These together imply that the procedure of $ \psi^E $ in $ \Gamma' $ remains the same with the procedure of $ \psi^E $ in $ \Gamma $ until all agents have exhausted their endowments in $ \Gamma $. Therefore, for all  $ i\in I $, if $ \w_i=\w'_i $, then $ \psi^E_i (\Gamma')=\psi^E_i (\Gamma) $; if $ \w_i\neq\w'_i $, then $ \psi^E_i (\Gamma')[\norm{\psi^E_i (\Gamma)}]= \psi^E_i (\Gamma)$.

(\textbf{Only if}) Let $ \psi $ be a mechanism satisfying sd-efficiency and endowment-expansion invariance. For any economy $ \Gamma =(I,O,\succ_I, \w)$, we prove that $ \psi(\Gamma)=\psi^E(\Gamma) $.

In the procedure of $ \psi^E $, for each step $ d $, we construct an economy $ \Gamma(d)=(I,O,\succ_I,\w^*(d)) $ such that, for all $ i\in I $, $ \w^*_i(d)=\w_i - \w_i(d+1) $. In words, $ \w^*_i(d) $ denotes $ i $'s exhausted endowments by the end of step $ d $. 
At each step, all remaining owners of each remaining object use equal amounts of the object for trading. Therefore, for all $ o\in O $ and all $ i,j\in I_o $ (where $ I_o $ denotes the set of owners of $ o $), at step one, $ \w^*_{i,o}(1)=\w^*_{j,o}(1) $; at any step $ d\ge 2 $, if $ \w^*_{i,o}(d)>\w^*_{i,o}(d-1)$, then, from step one to step $ d $, as long as $ o $ is involved in trading at any step, $ i $ contributes her endowment of $ o $ to trading at that step. So, $ \w^*_{i,o}(d-1)\ge \w^*_{j,o}(d-1) $ for all other $ j\in I_o\backslash \{i\} $, and for any other $ j\in I_o\backslash \{i\} $ that also satisfies $ \w^*_{j,o}(d)>\w^*_{j,o}(d-1)$, $ \w^*_{i,o}(d-1)=\w^*_{j,o}(d-1) $ and $ \w^*_{i,o}(d)= \w^*_{j,o}(d)$.

We prove by induction that, for all step $ d $, $ \psi(\Gamma(d))=\psi^E(\Gamma(d)) $. So, $ \psi(\Gamma)=\psi^E(\Gamma) $.

\textbf{Base step}: At step one of $ \psi^E $, every agent $ i $ supplies her endowment $ \w^*_i(1) $ in exchange for an equal amount of her favorite object. Therefore, $ \psi^E(\Gamma(1))  $ is the only sd-efficient allocation in the economy $ \Gamma(1) $. Since $ \psi $ is sd-efficient, we must have $ \psi(\Gamma(1))=\psi^E(\Gamma(1)) $. 

\textbf{Induction step}: Suppose that, for every $ 1\le d\le \ell $, $ \psi(\Gamma(d))=\psi^E(\Gamma(d))$. We want to prove that $ \psi(\Gamma(\ell+1))=\psi^E(\Gamma(\ell+1)) $. We first prove that $ \Gamma(\ell) $ and $ \Gamma(\ell+1) $ satisfy the conditions in the definition of endowment-expansion invariance. First, obviously, $ \w^*(\ell+1)\ge \w^*(\ell) $. Second, for all $ o\in O $ and all $ i\in I_o $ such that $ \w^*_{i,o}(\ell+1)> \w^*_{i,o}(\ell) $, we have explained that $ \w^*_{i,o}(\ell) \ge \w^*_{j,o}(\ell)$ for all $ j\in I_o\backslash \{i\} $. Last, for any object $ o $ such that $\sum_{i\in I} \w^*_{i,o}(\ell+1)> \sum_{i\in I} \w^*_{i,o}(\ell) $, it means that $ o $ is not exhausted at the end of step $ \ell $ of $ \psi^E $. Therefore, if $ \psi^E_{i,o'}(\Gamma(\ell))>0 $, $ i $ must most prefer $ o' $ in some step, which implies that $ o'\succsim_i o $. Then, endowment-expansion invariance of $ \psi $ implies that, for all $ i\in I $, if $ \w^*_{i}(\ell+1)=\w^*_{i}(\ell) $, then 	
$ \psi_i(\Gamma(\ell+1))= \psi_i(\Gamma(\ell))$, while if $ \w^*_i(\ell+1)\neq \w^*_i(\ell) $, then
$ \psi_i(\Gamma(\ell+1))[\norm{\psi_i(\Gamma(\ell))}]=\psi_i(\Gamma(\ell)) $. Since $ \psi $ satisfies sd-efficiency, $ \psi(\Gamma(\ell+1))- \psi(\Gamma(\ell))$ must be an sd-efficient allocation in the economy $(I,O,\succ_I, \w^*(\ell+1)-\w^*(\ell)) $. At step $ \ell+1 $ of $ \psi^E $, if an agent obtains a positive amount of any object, it must be her favorite object among the remaining ones. Therefore, $ \psi^E(\Gamma(\ell+1))- \psi^E(\Gamma(\ell))$ must be the only sd-efficient allocation in the economy $(I,O,\succ_I, \w^*(\ell+1)-\w^*(\ell)) $. So, $ \psi(\Gamma(\ell+1))- \psi(\Gamma(\ell))=\psi^E(\Gamma(\ell+1))- \psi^E(\Gamma(\ell))$, which, given $ \psi(\Gamma(\ell))=\psi^E(\Gamma(\ell)) $, implies that $ \psi(\Gamma(\ell+1))=\psi^E(\Gamma(\ell+1)) $. 
\end{proof}

\begin{proof}[\normalfont\textbf{Proof of Proposition \ref{prop:asymptotic:IC}}]

(\textbf{Bounded invariance}) Let $ \psi $ be a regular BTM. Consider an economy $ \Gamma= (I,O,\succ_I, \w) $. For any $ i\in I $, any $ o\in O $, and any $ \succ'_i \neq \succ_i $ such that $ \succ_i|_{U(\succsim_i,o)} = \succ'_i|_{U(\succsim'_i,o)} $, construct $ \Gamma' $ from $ \Gamma $ by replacing $ \succ_i $ with $ \succ'_i $. $ \Gamma $ and $ \Gamma' $ have equal endowment distributions. So, the procedures of $ \psi $ in the two economies coincide until $ o $ is exhausted. Therefore,  for every $ j\in I $, $\psi_{j,o}(\Gamma)=\psi_{j,o}(\Gamma')$.

(\textbf{Asymptotic strategy-proofness}) Let $ \psi $ be a continuous regular BTM. So, it satisfies EENE, which implies ETE. 
In any economy $ \Gamma=(I,O,\succ_I,\w) $, for every $ (\hat{\w},\succ)\in \Omega_\Gamma\times \mathcal{P} $, if there exist type-$ (\hat{\w},\succ) $ agents, they are treated equally in the procedure and outcome of $ \psi $. Let $ p =\psi (\Gamma) $. We use $ p_{(\hat{\w},\succ)} $ to denote the assignment received by each type-$ (\hat{\w},\succ) $ agent. Moreover, we describe the equations at each step of $ \psi $ by classifying students into types. 

Specifically, at each step $ d $ of $ \psi $, let $ \mathcal{T}(d) $ denote the set of agent types among $ I(d) $. For each $ (\hat{\w},\succ)\in \mathcal{T}(d) $, let $ N(\hat{\w},\succ)(d) $ and $ A(\hat{\w},\succ)(d) $ respectively denote the number and the proportion of type-$ (\hat{\w},\succ) $ agents among $ I(d) $; let $ o_{(\hat{\w},\succ)}(d) $ denote their favorite object among $ O(d) $; let $ \lambda_{(\hat{\w},\succ),o}(d) $ denote the parameter for them; let $ x_{(\hat{\w},\succ)}(d) $ denote the amount of the favorite object they receive at the step; let $ \w_{(\hat{\w},\succ)}(d) $ denote their endowment at the beginning of the step. Then, the equation system (\ref{equation7}) and constraints (\ref{equation9}) can be written as
\begin{equation*}
	\begin{cases}
		x_o(d)=\sum_{(\hat{\w},\succ)\in \mathcal{T}(d)} \mathbf{1}\{o_{(\hat{\w},\succ)}(d)=o\}\cdot N(\hat{\w},\succ)(d) \cdot x_{(\hat{\w},\succ)}(d)  & \text{for all } o\in O(d),\\
		x_{(\hat{\w},\succ)}(d)=\sum_{o\in O(d)} \lambda_{(\hat{\w},\succ),o}(d) \cdot x_o(d) &  \text{for all }(\hat{\w},\succ)\in \mathcal{T}(d),
	\end{cases}
\end{equation*}
subject to \quad
$
\lambda_{(\hat{\w},\succ),o}(d) x_o(d) \le \w_{(\hat{\w},\succ),o}(d) $ \quad for all $ (\hat{\w},\succ)\in \mathcal{T}(d) $ and all $ o\in O(d) $.

Define $ \hat{x}_o(d)= \frac{x_o(d)}{|I(d)|}$ and use it to replace $ x_o(d) $ in the above equations and constraints: 
\begin{equation}\label{equation10}
	\begin{cases}
		\hat{x}_o(d)=\sum_{(\hat{\w},\succ)\in \mathcal{T}(d)} \mathbf{1}\{o_{(\hat{\w},\succ)}(d)=o\}\cdot A(\hat{\w},\succ)(d) \cdot x_{(\hat{\w},\succ)}(d)  & \text{for all } o\in O(d),\\
		x_{(\hat{\w},\succ)}(d)=\sum_{o\in O(d)} \lambda_{(\hat{\w},\succ),o}(d) \cdot |I(d)| \cdot \hat{x}_o(d) &  \text{for all }(\hat{\w},\succ)\in \mathcal{T}(d),
	\end{cases}
\end{equation}
subject to \quad 
$
\lambda_{(\hat{\w},\succ),o}(d) \cdot |I(d)| \cdot \hat{x}_o(d) \le \w_{(\hat{\w},\succ),o}(d) $ \quad  for all $ (\hat{\w},\succ)\in \mathcal{T}(d) $ and all $ o\in O(d) $.

Let $ \Theta(d) $ denote the coefficient matrix of the equation system (\ref{equation10}) where the rows and columns are indexed by $ O(d) \cup \mathcal{T}(d) $. The elements of $ \Theta(d) $ are $ A(\hat{\w},\succ)(d) $ and $ \lambda_{(\hat{\w},\succ),o}(d) |I(d)|  $.

Now, consider any regular sequence of economies $ (\Gamma^{[n]})_{n=1}^\infty $, where $ \Gamma^{[1]} =\Gamma$. There exists $ N^*>0 $ such that, for all $ n\ge N^* $, there are more than one type-$ (\hat{\w},\succ) $ agent for all $(\hat{\w},\succ)\in \Omega_\Gamma\times \mathcal{P} $ in $ \Gamma^{[n]} $. Suppose for some $ (\hat{\w},\succ)\in \Omega_\Gamma\times \mathcal{P} $, a type-$ (\hat{\w},\succ) $ agent reports $ \succ' \neq \succ $ in all $ \Gamma^{[n]} $ with $ n \ge N^* $. Then, we obtain a new regular sequence of economies, $ (\tilde{\Gamma}^{[n]})_{n=N^*}^\infty $. For every $ n\ge N^* $, $\tilde{N}^{[n]}(\hat{\w},\succ)=N^{[n]}(\hat{\w},\succ)-1
$, $\tilde{N}^{[n]}(\hat{\w},\succ')=N^{[n]}(\hat{\w},\succ')+1
$, and the numbers of agents of other types in $ \tilde{\Gamma}^{[n]} $ are equal to those in $ \Gamma^{[n]} $. However, for all $ (\hat{\w},\succ)\in \Omega_\Gamma\times \mathcal{P} $, $ \lim_{n \rightarrow \infty}\tilde{A}^{[n]}(\hat{\w},\succ)=\lim_{n \rightarrow \infty} A^{[n]}(\hat{\w},\succ)  $.

Let $ (p^{[n]})_{n=N^*}^\infty $ and $ (\tilde{p}^{[n]})_{n=N^*}^\infty $ denote the sequences of allocations found by $ \psi $ for $ (\Gamma^{[n]})_{n=N^*}^\infty $ and $ (\tilde{\Gamma}^{[n]})_{n=N^*}^\infty $, respectively. For every $ n \ge N^* $, $ \tilde{p}^{[n]} $ may differ from $ p^{[n]} $. However, we prove that, because $ \psi $ is continuous and $ \lim_{n \rightarrow \infty}\tilde{A}^{[n]}(\hat{\w},\succ)=\lim_{n \rightarrow \infty} A^{[n]}(\hat{\w},\succ)  $ for all $ (\hat{\w},\succ) $, $$ \lim_{n \rightarrow \infty}(\tilde{p}^{[n]}-p^{[n]}) =\mathbf{0}. $$

Specifically, let $ \Theta^{[n]}(d) $ and $ \tilde{\Theta}^{[n]}(d) $ denote the coefficient matrix of the equation system at step $ d $ of $ \psi $ in $ \Gamma^{[n]} $ and $ \tilde{\Gamma}^{[n]} $, respectively. 

At step one of $ \psi $, because agents' endowments are same in $ \Gamma^{[n]} $ and $ \tilde{\Gamma}^{[n]} $, if $ \Theta^{[n]}(d) $ differs from $ \tilde{\Theta}^{[n]}(d) $, the difference must be due to the manipulation. That is, $ A^{[n]}(\hat{\w},\succ) > \tilde{A}^{[n]}(\hat{\w},\succ)$ and $ A^{[n]}(\hat{\w},\succ') < \tilde{A}^{[n]}(\hat{\w},\succ')$. However, since $ \lim_{n \rightarrow \infty}\tilde{A}^{[n]}(\hat{\w},\succ)=\lim_{n \rightarrow \infty} A^{[n]}(\hat{\w},\succ)  $ for all $ (\hat{\w},\succ) $, as $ n $ goes to infinity, $ \Theta^{[n]}(1) $ and $ \tilde{\Theta}^{[n]}(1) $ are arbitrarily close, implying that the maximum solutions at step one of $ \psi $ in the two economies are arbitrarily close. 

Then, consider step two of $ \psi $. Since the maximum solutions at step one of $ \psi $ in the two economies may differ, agents' endowments at the beginning of step two in the two economies may differ. However, as $ n $ goes to infinity, since the maximum solutions at step one in the two economies are arbitrarily close, agents' endowments at step two in the two economies must also be arbitrarily close. Then, because $ \psi $ is continuous, $ \lambda_{(\hat{\w},\succ),o}(2)  $ and $ \tilde{\lambda}_{(\hat{\w},\succ),o}(2) $ for all $ (\hat{\w},\succ) $ and all $ o $ in the two economies must be arbitrarily close. Therefore, $ \Theta^{[n]}(2) $ and $ \tilde{\Theta}^{[n]}(2) $ are arbitrarily close, implying that their maximum solutions at step two are arbitrarily close. 

The above argument can be inductively applied to the remaining steps of $ \psi $. Because there are finite agent types, the procedure of $ \psi $ must stop in finite steps with a fixed upper bound in all economies. Therefore, we conclude that $ \lim_{n \rightarrow \infty}(\tilde{p}^{[n]}-p^{[n]}) =\mathbf{0}. $ This means that, for any $ \varepsilon>0 $, there exists $ n^*(\hat{\w},\succ,\succ')>N^* $ such that, for all $ n> n^*(\hat{\w},\succ,\succ')$,
\[
\max_{o\in O}\big|\sum_{o'\succsim o}p^{[n]}_{(\hat{\w},\succ'),o'}-\sum_{o'\succsim o}\tilde{p}^{[n]}_{(\hat{\w},\succ'),o'}\big|\le \varepsilon.
\]

Since $ \psi $ satisfies EENE, $ p^{[n]}_{(\hat{\w},\succ)} \succsim^{sd}p^{[n]}_{(\hat{\w},\succ')} $.
That is, for all $ o\in O $, 
$
\sum_{o'\succsim o} p^{[n]}_{(\hat{\w},\succ),o'} \ge \sum_{o'\succsim o} p^{[n]}_{(\hat{\w},\succ'),o'}
$. Thus, for all $ o\in O $, \[
\sum_{o'\succsim o} p^{[n]}_{(\hat{\w},\succ),o'} \ge  \sum_{o'\succsim o}\tilde{p}^{[n]}_{(\hat{\w},\succ'),o'}-\varepsilon.
\]

Define $ n^*=\max\{n^*(\hat{\w},\succ,\succ'):(\hat{\w},\succ,\succ')\in \Omega_{\Gamma}
\times \mathcal{P}\times \mathcal{P}, \succ\neq \succ' \} $. Then, for all $ n> n^* $ and all $ i\in I^{[n]}$  in $ \Gamma^{[n]} $, if $ i $ reports any $ \succ'\neq \succ $ when her true preference is $ \succ $, then her assignment from truth-telling weakly stochastically dominates her assignment from reporting $ \succ' $, with an error bounded by $ \varepsilon $.
\end{proof}

\begin{proof}[\normalfont\textbf{Proof of Proposition \ref{prop:priority}}]
Sd-efficiency and bounded invariance hold similarly as for $ \psi^E $.  If an agent $ i $ has weakly higher priority than another $ j $ for every object, at every step of PTM, $ i $ obtains weakly more of her favorite object than $ j $ does. So, $ i $ does not envy $ j $.
\end{proof}

\begin{proof}[\normalfont\textbf{Proof of Proposition \ref{prop:pse}}]
We apply PTM to the HET model by regarding the model as a special case of the priority-based allocation model, and explain that the procedure essentially coincides with that of ETM. 

Let $ I(d) $, $ O(d) $, and $ o_i(d) $ be defined as in PTM. Among $ O(d) $, let $ O_E(d)=\{o\in O_E:i_o\in I(d)\} $, which is the set of private endowments whose owners have not been removed at the beginning of step $ d $. Then, at each step $ d $, PTM finds the maximum solution $ \mathbf{x}^*(d) $ to
\begin{equation*}
	\begin{cases}
		x_o(d)=\sum_{i\in I(d)} \mathbf{1}\{o_i(d)=o\} \cdot x_i(d) & \text{ for all } o\in O(d),\\ 
		x_i(d)=\sum_{o\in O(d)\backslash O_E(d)} \dfrac{x_o(d)}{|I(d)|} & \text{ for all }i\in I(d)\backslash I_E,\\
		x_j(d)=\sum_{o\in O(d)\backslash O_E(d)} \dfrac{x_o(d)}{|I(d)|}+x_{o_j}(d) & \text{ for all }j\in I(d)\cap I_E.
	\end{cases}
\end{equation*}
subject to the constraints
\begin{equation*}
	\begin{cases}
		x_o(d)\le 1-\sum_{k=1}^{d-1} x^*_o(k) & \text{ for all } o\in O(d),\\
		x_i(d)\le 1-\sum_{k=1}^{d-1} x^*_i(k) & \text{ for all }i\in I(d).
	\end{cases}
\end{equation*}

At each step $ d $ of ETM, let $ y_i(d) $ denote the amount of the favorite object received by each $ i\in I(d) $; let $ y_o(d) $ denote the amount of each $ o\in O(d)  $ assigned to agents. We show that $ y(d)= \{y_a(d)\}_{a\in I(d)\cup O(d)} $  is the maximum solution to the above equations.

At each step $ d $ of ETM, if there exist cycles among existing tenants, for each existing tenant $ j $ involved in a cycle, it is obvious that $ y_j(d)=y_{\w(j)}(d) $, whereas for each agent $ i $ and each object $ o $  not involved in any cycle, $ y_i(d)=y_o(d)=0 $. So, $ y(d) $ satisfies the equations and constraints at step $ d $ of PTM. 

If there do not exist cycles, let $ t(d) $ be the duration of step $ d $. Then, for each $ o\in O(d) $, $ y_o(d)=\sum_{i\in I(d)} \mathbf{1}\{o_i(d)=o\} \cdot y_i(d) $. For each $ i\in I(d)\backslash I_E $, $ y_i(d)=s_i(d)t(d)=t(d) $, whereas for each $ j\in I(d)\cap I_E $, $ y_j(d)=s_j(d)t(d)=\sum_{i\in I(d)}\mathbf{1}\{o_i(d)=\w(j)\}\cdot s_i(d)t(d)+t(d)=\sum_{i\in I(d)}\mathbf{1}\{o_i(d)=\w(j)\}\cdot y_i(d)+t(d)=y_{\w(j)}(d)+t(d)$. We say that an agent $ i $ is linked to an object $ o $ if there exist distinct existing tenants $ j_1,j_2,\ldots, j_\ell $ such that $ i $ points to $ \w(j_1) $, $ j_1 $ points to $ \w(j_2) $, $ j_2 $ points to $ \w(j_3) $, \ldots, $ j_\ell $ points to $ o $. Because there do not exist cycles, every agent in $ I(d) $ must be linked to an object in $ O(d)\backslash O_E(d) $. Then, the ``you request my house - I get your rate'' rule implies that the sum of the rates of the agents who are eating the objects in $ O(d)\backslash O_E(d) $ equals the total number of agents, $ \norm{I(d)} $. So, $|I(d)| t(d)=\sum_{o\in O(d)\backslash O_E(d)} y_o(d)$. Therefore, for each $ i\in I(d)\backslash I_E $, $ y_i(d)= \sum_{o\in O(d)\backslash O_E(d)} y_o(d)/\norm{I(d)}$, whereas for each $ j\in I(d)\cap I_E $, $ y_j(d)=y_{\w(j)}(d)+\sum_{o\in O(d)\backslash O_E(d)} y_o(d)/\norm{I(d)}$. This means that $ y(d) $ satisfies the equations and constraints at step $ d $ of PTM. 

Each step $ d $ of ETM stops when some agent is satisfied or when some object is exhausted. Therefore, $ y(d) $ must coincide with the maximum solution to the equations at step $ d $ of PTM. This means that ETM is equivalent to the reduction of PTM to the HET model.

Then, the sd-efficiency, no envy towards newcomers, and bounded invariance of ETM are implications of the properties of PTM in Proposition \ref{prop:priority}. ETM is IR for existing tenants because every private endowment is never exhausted before its owner being removed.
\end{proof}


\section{The procedure of $ \psi^E $ in Example \ref{example:2}}\label{section:BTM:procedure}

Step 1 of $ \psi^E $ in Example \ref{example:2} has been presented in Section \ref{section:method}. This appendix presents the remaining steps of $ \psi^E $. \autoref{figure:procedure:Example2} draws the generated graphs at all steps.

\begin{description}

\item[\textbf{Step 2:}] $ 1,2, 3, 5$ most prefer $ c $ and $ 4 $ most prefers $ a $. We solve the following equations: 
\begin{equation*}
	\scalebox{.9}{$\begin{cases}
			x_{a}=x_4,\\
			x_{b}=0,\\
			x_{c}=x_1+x_2+x_3+x_5,\\
			x_{e}=0,\\
			x_1=1/2x_{a} +1/2x_{b},\\
			x_2=1/2x_{a} +1/2x_{b}\\
			x_3=1/3x_{c} +1/3x_{e},\\
			x_4=1/3x_{c} +1/3x_{e},\\
			x_5=1/3x_{c} +1/3x_{e}.
		\end{cases}
		\text{subject to} \quad
		\begin{cases}
			1/2x_{a}\le 1/6,\\
			1/2x_{b}\le 1/2, \\
			1/3x_{c}\le 1/12, \\
			1/3x_{e}\le 1/4.
		\end{cases}$
	}
\end{equation*}

The maximum solution is
\[
\scalebox{.9}{$
	\mathbf{x}^*=\left(
	\begin{array}{ccccccccc}
		x^*_{1} & x^*_{2} & x^*_{3} & x^*_{4} &x^*_{5} & x^*_{a} & x^*_{b} & x^*_{c} & x^*_{e}\\
		1/24&1/24&1/12&1/12&1/12&1/12&0&3/12&0
	\end{array}
	\right).$
}
\]

The amounts of $ c $ owned by $ 3 $ and $ 4 $ are exhausted. But $ 5 $ still owns $ 1/4c $.

\item[\textbf{Step 3:}] We solve the following equations:

\begin{equation*}
	\scalebox{.9}{$\begin{cases}
			x_{a}=x_4,\\
			x_{b}=0,\\
			x_{c}=x_1+x_2+x_3+x_5,\\
			x_{e}=0,\\
			x_1=1/2x_{a} +1/2x_{b},\\
			x_2=1/2x_{a} +1/2x_{b}\\
			x_3=1/3x_{e},\\
			x_4=1/3x_{e},\\
			x_5=x_{c} +1/3x_{e}.
		\end{cases}
		\text{subject to} \quad
		\begin{cases}
			1/2x_{a}\le 1/8,\\
			1/2x_{b}\le 1/2, \\
			x_{c}\le 1/4, \\
			1/3x_{e}\le 1/4.
		\end{cases}$
	}
\end{equation*}

In the maximum solution, only $ 5 $ obtains her endowment $ 1/4c $. Then, $ c $ is exhausted.

\item[\textbf{Step 4:}] $ 1,2,3,4 $ most prefer $ a $ and $ 5 $ most prefers $ e$. We solve the following equations:

\begin{equation*}
	\scalebox{.9}{$\begin{cases}
			x_{a}=x_1+x_2+x_3+x_4,\\
			x_{b}=0,\\
			x_{e}=x_5,\\
			x_1=1/2x_{a} +1/2x_{b},\\
			x_2=1/2x_{a} +1/2x_{b}\\
			x_3=1/3x_{e},\\
			x_4=1/3x_{e},\\
			x_5=1/3x_{e}.
		\end{cases}
		\text{subject to} \quad
		\begin{cases}
			1/2x_{a}\le 1/8,\\
			1/2x_{b}\le 1/2, \\
			1/3x_{e}\le 1/4.
		\end{cases}$
	}
\end{equation*}

In the maximum solution,  $ 1$ and $2$ obtain their endowment $ 1/8a $. Then, $ a $ is exhausted.

\item[\textbf{Step 5:}]  $ 1,2 $ most prefer $ b $ and $ 3,4,5 $ most prefer $ e $. We solve the following equations:

\begin{equation*}
	\scalebox{.9}{$\begin{cases}
			x_{b}=x_1+x_2,\\
			x_{e}=x_3+x_4+x_5,\\
			x_1=1/2x_{b},\\
			x_2=1/2x_{b}\\
			x_3=1/3x_{e},\\
			x_4=1/3x_{e},\\
			x_5=1/3x_{e}.
		\end{cases}
		\text{subject to} \quad
		\begin{cases}
			1/2x_{b}\le 1/2, \\
			1/3x_{e}\le 1/4.
		\end{cases}$
	}
\end{equation*}

In the maximum solution, $ 1 $ and $2 $ obtain their endowment $ 1/2b $; $ 3 $, $ 4 $, and $ 5 $ obtain their endowment $ 1/4e $.

\item[\textbf{Step 6:}] $ 5 $ is the only remaining agent who owns the remaining endowment $ 1/4e $. So, $ 5 $ obtains $ 1/4e $. The graph for this step is omitted in \autoref{figure:procedure:Example2}.		
\end{description}

The allocation found by $ \psi^E $ is 

\begin{table}[!htb]
\centering
\scalebox{.9}{
	\begin{tabular}[c]{c|ccccc}
		& $a$ & $b$ & $c$ & $d$ & $e$ \\\hline
		$1$ & $1/8$ & $1/2$ & $3/8$ &  & \\
		$2$ & $1/8$ & $1/2$ & $1/24$ & $1/3$ &  \\
		$3$ &  &  & $1/12$ & $2/3$ & $1/4$ \\
		$4$ & $3/4$ &  &  &  & $1/4$ \\
		$5$ &  &  & $1/2$ &  & $1/2$ \\
	\end{tabular}
}
\end{table}

\begin{figure}[!ht]
\centering
\scalebox{.9}{
	\begin{subfigure}{.45\linewidth}
		\raggedright
		\begin{tikzpicture}[semithick,bend angle=20,xscale=0.8,yscale=0.8]
			\node (o1) at (2,0) [object] {$a$};
			\node (o5) at (8,5) [object] {$e$};
			\node (o3) at (5,2) [object] {$c$};
			\node (o2) at (0,0) [object] {$b$};
			\node (i1) at (5,0) [agent] {$1$};
			\node (i2) at (0,5) [agent] {$2$};
			\node (i3) at (5,5) [agent] {$3$};
			\node (i4) at (2,2) [agent] {$4$};
			\node (i5) at (8,2) [agent] {$5$};
			
			\node at (0.8,2) {\footnotesize $1/6$};
			\node at (3.5,1.7) {\footnotesize $1/12$};
			\node at (4.2,3.5) {\footnotesize $1/12$};
			\node at (6.6,4) {\footnotesize $1/4$};
			\node at (6.5,5.3) {\footnotesize $1/4$};
			\node at (6.5,1.4) {\footnotesize $1/3$};
			\node at (7.7,3.5) {\footnotesize $1/2$};
			\node at (-.3,2) {\footnotesize $1/2$};
			\node at (2.4,-.8) {\footnotesize $1/2$};
			\node at (3.5,.2) {\footnotesize $1/6$};
			
			\draw[-latex] (i1) to (o3);
			\draw[-latex] (i2) to (o3);
			\draw[-latex] (i3) [bend left] to (o3);
			\draw[-latex] (i5) [bend right] to (o3);
			\draw[-latex] (i4) to (o1);
			\draw[-latex] (o1) to (i1);
			\draw[-latex] (o1) to (i2);
			\draw[-latex] (o3) [bend left] to (i3);
			\draw[-latex] (o3) to (i4);
			\draw[-latex] (o3) [bend right] to (i5);
			\draw[-latex,dashed] (o2) to (i2);
			\draw[-latex,dashed] (o2) [bend right] to (i1);
			\draw[-latex,dashed] (o5) to (i3);
			\draw[-latex,dashed] (o5) to (i5);
			\draw[-latex,dashed] (o5) to (i4);
		\end{tikzpicture}
		\subcaption{Step 2}\label{figure:step2}
	\end{subfigure}
	\quad\quad\quad
	\begin{subfigure}{.45\linewidth}
		\raggedleft
		\begin{tikzpicture}[semithick,bend angle=20,xscale=0.8,yscale=0.8]
			\node (o1) at (2,0) [object] {$a$};
			\node (o5) at (8,5) [object] {$e$};
			\node (o3) at (5,2) [object] {$c$};
			\node (o2) at (0,0) [object] {$b$};
			\node (i1) at (5,0) [agent] {$1$};
			\node (i2) at (0,5) [agent] {$2$};
			\node (i3) at (5,5) [agent] {$3$};
			\node (i4) at (2,2) [agent] {$4$};
			\node (i5) at (8,2) [agent] {$5$};
			
			\node at (0.8,2) {\footnotesize $1/8$};
			\node at (6.6,4) {\footnotesize $1/4$};
			\node at (6.5,5.3) {\footnotesize $1/4$};
			\node at (6.5,1.4) {\footnotesize $1/4$};
			\node at (7.7,3.5) {\footnotesize $1/2$};
			\node at (-.3,2) {\footnotesize $1/2$};
			\node at (2.4,-.8) {\footnotesize $1/2$};
			\node at (3.5,.2) {\footnotesize $1/8$};
			
			\draw[-latex,dashed] (i1) to (o3);
			\draw[-latex,dashed] (i2) to (o3);
			\draw[-latex,dashed] (i3)  to (o3);
			\draw[-latex] (i5) [bend right] to (o3);
			\draw[-latex,dashed] (i4) to (o1);
			\draw[-latex,dashed] (o1) to (i1);
			\draw[-latex,dashed] (o1) to (i2);
			\draw[-latex] (o3) [bend right] to (i5);
			\draw[-latex,dashed] (o2) to (i2);
			\draw[-latex,dashed] (o2) [bend right] to (i1);
			\draw[-latex,dashed] (o5) to (i3);
			\draw[-latex,dashed] (o5) to (i5);
			\draw[-latex,dashed] (o5) to (i4);
		\end{tikzpicture}
		\subcaption{Step 3}\label{figure:step3}
	\end{subfigure}
}
\par\bigskip
\scalebox{.9}{
	\begin{subfigure}{.45\linewidth}
		\raggedright
		\begin{tikzpicture}[semithick,bend angle=20,xscale=0.8,yscale=0.8]
			\node (o1) at (2,2) [object] {$a$};
			\node (o5) at (8,5) [object] {$e$};
			\node (o2) at (0,0) [object] {$b$};
			\node (i1) at (5,0) [agent] {$1$};
			\node (i2) at (0,5) [agent] {$2$};
			\node (i3) at (5,5) [agent] {$3$};
			\node (i4) at (5,2) [agent] {$4$};
			\node (i5) at (8,2) [agent] {$5$};
			
			\node at (.85,3.4) {\footnotesize $1/8$};
			\node at (6.5,3.1) {\footnotesize $1/4$};
			\node at (6.5,4.7) {\footnotesize $1/4$};
			\node at (8.7,3.5) {\footnotesize $1/2$};
			\node at (-.4,2) {\footnotesize $1/2$};
			\node at (2.4,-.4) {\footnotesize $1/2$};
			\node at (2.8,.6) {\footnotesize $1/8$};
			
			\draw[-latex] (i1) [bend right] to (o1);
			\draw[-latex] (i2) [bend left] to (o1);
			\draw[-latex,dashed] (i3)  to (o1);
			\draw[-latex,dashed] (i5) [bend left] to (o5);
			\draw[-latex,dashed] (i4) to (o1);
			\draw[-latex] (o1) [bend right] to (i1);
			\draw[-latex] (o1) [bend left] to (i2);
			\draw[-latex,dashed] (o2) to (i2);
			\draw[-latex,dashed] (o2)  to (i1);
			\draw[-latex,dashed] (o5) to (i3);
			\draw[-latex,dashed] (o5) [bend left] to (i5);
			\draw[-latex,dashed] (o5)  to (i4);
		\end{tikzpicture}
		\subcaption{Step 4}\label{figure:step4}
	\end{subfigure}
	\quad\quad \quad
	\begin{subfigure}{.45\linewidth}
		\vspace{1.6cm}
		\raggedleft
		\begin{tikzpicture}[semithick,bend angle=15,xscale=.8,yscale=.8]
			\node (i1) at (0,0) [agent] {$1$};
			\node (i2) at (0,3) [agent] {$2$};
			\node (i3) at (4,3) [agent] {$3$};
			\node (i4) at (4,0) [agent] {$4$};
			\node (i5) at (8,1.5) [agent] {$5$};
			\node (o5) at (6,1.5) [object] {$e$};
			\node (o2) at (2,1.5) [object] {$b$};		
			
			\node at (1.4,2.6) {\footnotesize $1/2$};
			\node at (1.4,.3) {\footnotesize $1/2$};
			\node at (5.4,2.6) {\footnotesize $1/4$};
			\node at (5.4,.3) {\footnotesize $1/4$};
			\node at (7,1) {\footnotesize $1/2$};
			
			\draw[-latex] (i1) [bend right] to (o2);
			\draw[-latex] (i2) [bend right] to (o2);
			\draw[-latex] (i3) [bend right] to (o5);
			\draw[-latex] (i4) [bend right] to (o5);
			\draw[-latex] (i5) [bend right] to (o5);
			\draw[-latex] (o2) [bend right] to (i1);
			\draw[-latex] (o2) [bend right] to (i2);
			\draw[-latex] (o5) [bend right] to (i3);
			\draw[-latex] (o5) [bend right] to (i4);
			\draw[-latex] (o5) [bend right] to (i5);
		\end{tikzpicture}
		\subcaption{Step 5}\label{figure:step5}
	\end{subfigure}
}
\caption{In each graph, the nodes connected by solid arrows constitute an absorbing set.}\label{figure:procedure:Example2}
\end{figure}

\clearpage

\pagenumbering{arabic}

\begin{center}
\LARGE Online Appendix \\\normalsize(Not for publication)	  
\end{center}

\section{Proof of Lemma \ref{thm:existence}}\label{appendix_Theorem1}

The equation system $ \Lambda \bx=\bx $ can be written as $ (\mathbf{I}-\Lambda)\mathbf{x}=\mathbf{0}$.
Because every column of $ \Lambda $ sums to one, the sum of all columns of $ \mathbf{I}-\Lambda $ is zero. So $ \mathbf{I}-\Lambda $ is singular. It means that $ (\mathbf{I}-\Lambda)\mathbf{x}=\mathbf{0}$ has nonzero solutions. We want to prove that it has nonnegative solutions, and it has a maximum solution that satisfies the constraints $ \bx\le \mathbf{q} $.

We introduce some definitions. In a directed graph $ (\mathcal{V},\mathcal{E}) $, for any $ V\subseteq \cv $, let $ \Lambda_{V}=(\lambda_{v,u})_{v,u\in V} $ be the restriction of  $ \Lambda $ to $ V $. For any $ V'\subseteq V $, $ \Lambda_{V'} $ is called a \textit{submatrix} of $ \Lambda_{V} $; if $ V'\subsetneq V $, $ \Lambda_{V'} $ is called a strict submatrix of $ \Lambda_{V} $. $ \Lambda_{V'} $ is called \textit{isolated in $ \Lambda_{V} $} if $ \lambda_{v,u}=0 $ for all $ v\in V\backslash V'  $ and $ u\in V' $. For any $ V $, $ \Lambda_{V} $ is called \textit{isolated} if it is isolated in $ \Lambda $; otherwise it is called \textit{unisolated}. $ \Lambda_{V} $ is called \textit{reducible} if it has a strict nonempty submatrix that is isolated in $ \Lambda_{V} $; otherwise, it is called \textit{irreducible}. For convenience, we also say $ V $ is (un)isolated or (ir)reducible if $ \Lambda_{V} $ is (un)isolated or (ir)reducible.

\begin{lemma}\label{lemma1}
$ \cv $ can be uniquely partitioned into disjoint sets $ V_1, V_2,\ldots,V_k,V_{k+1} $ such that

(1) for all $ \ell=1,\ldots,k $, $ V_\ell $ is nonempty, isolated, and irreducible;

(2) $ V_{k+1} $, which can be empty, is either unisolated or reducible, and it does not contain any strict nonempty subset that is isolated and irreducible.
\end{lemma}

\begin{proof}

We prove the lemma by constructing the partition. We first construct $ V_1 $. If $ \cv $ is irreducible, let $ V_1=\cv $, and we are done by letting $ k=1 $ and $ V_2=\emptyset $. 	
Otherwise, $ \cv $ contains a strict nonempty subset $ V $ that is isolated. If $ V $ is irreducible, let $ V_1=V  $. Otherwise, $ V $ contains a strict nonempty subset $ V' $ that is isolated in $ V $. Since $ V $ is isolated, $ V' $ is isolated. If $ V' $ is irreducible, let $ V_1=V' $. Otherwise, $ V' $ also contains a strict nonempty subset $ V'' $ that is isolated in $ V' $. So $ V'' $ is also isolated. Note that every submatrix that consists of a single element must be irreducible. So by repeating the above argument we must be able to find an irreducible and isolated nonempty subset. Let the subset be $ V_1 $. 

We then construct $ V_2 $. 	If $ \cv\backslash V_1 $ is irreducible and isolated, we are done by letting $ V_2=\cv\backslash V_1 $, $ k=2 $, and $ V_3=\emptyset $.
Otherwise, if $ \cv\backslash V_1 $ does not contain any strict nonempty subset that is irreducible and isolated, we are done by letting $ k=1 $ and $ V_2=\cv\backslash V_1 $. If $ \cv\backslash V_1 $ contains a strict nonempty subset $ V $ that is irreducible and isolated, let $ V_2=V $.
Then we construct $ V_3 $ and possibly $ V_4,\ldots, V_{k+1} $ from $ \cv\backslash (V_1\cup V_2) $ in a similar way as above. Since $ \cv $ is a finite set, we must stop in finite steps.
\end{proof}
Lemma \ref{lemma1} implies that by permuting rows and columns, we can write $ \Lambda $ as the following block form:
\begin{equation}\label{blockform}
\scalebox{.9}{$
	\Lambda = \left[
	\begin{array}{cccc|c}
		\multicolumn{1}{c}{\Lambda_{V_1}} & & & \multirow{2}{.5cm}{\large \textbf{0}} & \multirow{3}{.5cm}{\large Z} \\
		& \multicolumn{1}{c}{\Lambda_{V_2}} & & & \\
		&  & \ddots & & \\
		\multirow{2}{0cm}{\large \textbf{0}} & &  & \multicolumn{1}{c|}{\Lambda_{V_k}} & \\ \cline{5-5}
		& & & & \Lambda_{{V_{k+1}}} \\
	\end{array}
	\right]
	$}
\end{equation}

\bigskip
For every $ \ell=1,\ldots,k+1 $, let $ \mathbf{I}_{V_\ell} $ denote the identity matrix of dimension $ |V_\ell|\times |V_\ell| $.
\begin{lemma}\label{lemma2}
(1) $ \text{Rank}(\mathbf{I}_{V_\ell}-\Lambda_{V_\ell})=|V_\ell|-1 $ for all $ \ell=1,\ldots,k $;

(2) $ \text{Rank}(\mathbf{I}_{V_{k+1}}-\Lambda_{V_{k+1}})=|V_{k+1}| $.
\end{lemma}

\begin{proof}
(1) For every $ \ell=1,\ldots,k $, since $ V_\ell $ is isolated, it must be that $ \sum_{v\in V_\ell} \lambda_{v,u}=1 $ for every $ u\in V_\ell $. Therefore, the sum of all columns of $ \mathbf{I}_{V_\ell}-\Lambda_{V_\ell} $ is zero. So $ \text{det}(\mathbf{I}_{V_\ell}-\Lambda_{V_\ell})=0 $ and $ \text{Rank}(\mathbf{I}_{V_\ell}-\Lambda_{V_\ell})<|V_\ell| $.

\medskip

\textbf{Claim 1.} For every nonempty $ V\subsetneq V_\ell $, there exists $ u\in V $ such that $ \sum_{v\in V}\lambda_{v,u}<1 $.

\textbf{Proof}: Suppose $ \sum_{v\in V}\lambda_{v,u}=1 $ for all $ u\in V $. Since $ \sum_{v\in V_\ell} \lambda_{v,u}=1 $ for every $ u\in V_\ell $, it means that $ \lambda_{v,u}=0 $ for all $ v\in V_\ell \backslash V $ and $ u\in V $. So $ \Lambda_{V} $ is isolated in $ \Lambda_{V_\ell} $, which contradicts the fact that $ \Lambda_{V_\ell} $ is irreducible.

\medskip

Corollary 3.3 of \cite{peterson1982leontief} states that for a general matrix $ \mathbf{D}=(d_{ij})_{i,j=1}^k $ such that $ d_{ij}\in [0,1] $ and $ \sum_{i=1}^k d_{ij}\le 1 $ for all $ j\in X=\{1,\ldots, k\} $, if $ \text{det}(\mathbf{I}_{k\times k}-\mathbf{D})=0 $ and $ \text{det}(\mathbf{I}_{(k-1)\times (k-1)}-\mathbf{D}_{X\backslash {j}})\neq 0 $ for all $ j\in X $, then every column of $ \mathbf{D} $ sums to one. This result and Claim 1 imply the following claim. 

\medskip

\textbf{Claim 2.} For every nonempty $ V\subsetneq V_\ell $, if $ \text{det}(\mathbf{I}_{V}-\Lambda_{V})=0 $, then there exists some $ u\in V $ such that $ \text{det}(\mathbf{I}_{V\backslash \{u\}}-\Lambda_{V\backslash \{u\}})=0 $.

\textbf{Proof}: Suppose $ \text{det}(\mathbf{I}_{V\backslash \{u\}}-\Lambda_{V\backslash \{u\}})\neq 0 $ for all $ u\in V $. By Corollary 3.3 of \cite{peterson1982leontief}, every column of $ \Lambda_{V} $ sums to one. But it contradicts Claim 1.

\medskip

Now we prove that for all $ u\in V_\ell $, $ \text{det}(\mathbf{I}_{V_\ell\backslash \{u\}}-\Lambda_{V_\ell\backslash \{u\}})\neq 0 $. Suppose towards a contradiction that $ \text{det}(\mathbf{I}_{V_\ell\backslash \{u\}}-\Lambda_{V_\ell\backslash \{u\}})= 0 $ for some $ u\in V_\ell $. By Claim 2, there exists $ u_1\in V_\ell\backslash \{u\} $ such that $ \text{det}(\mathbf{I}_{V_\ell\backslash \{u,u_1\}}-\Lambda_{V_\ell\backslash \{u,u_1\}})= 0 $. By Claim 2 again, there further exists $ u_2\in V_\ell\backslash \{u,u_1\} $ such that $ \text{det}(\mathbf{I}_{V_\ell\backslash \{u,u_1,u_2\}}-\Lambda_{V_\ell\backslash \{u,u_1,u_2\}})= 0 $. By repeatedly using Claim 2, we must find a submatrix consisting of only one element $ v\in V_\ell $ such that $ 1-\lambda_{v,v}=0 $, which contradicts the fact that $ \lambda_{v,v}=0 $ for all $ v\in \cv $. So $ \text{det}(\mathbf{I}_{V_\ell\backslash \{u\}}-\Lambda_{V_\ell\backslash \{u\}})\neq 0 $ for all $ u\in V_\ell $. This implies that $ \text{Rank}(\mathbf{I}_{V_\ell}-\Lambda_{V_\ell})=|V_\ell|-1 $.  
\medskip

(2) We first prove that $ \Lambda_{V_{k+1}} $ is unisolated. Suppose it is isolated, then by definition it must be reducible. So $ \Lambda_{V_{k+1}} $ contains a strict nonempty submatrix $ \Lambda_{V} $ that is isolated in $ \Lambda_{V_{k+1}}  $. Since $ \Lambda_{V_{k+1}} $ is isolated, $ \Lambda_{V} $ is also isolated. Still by the definition of $ \Lambda_{V_{k+1}} $, $ \Lambda_{V} $ must be reducible. So $ \Lambda_{V} $ also contains a strict nonempty submatrix $ \Lambda_{V'} $ that is isolated in $ \Lambda_{V} $. Since $ \Lambda_{V} $ is isolated, $ \Lambda_{V'} $ is also isolated. By the definition of $ \Lambda_{V_{k+1}} $ again, $ \Lambda_{V'} $ must be reducible. By repeating this argument, we will finally find a submatrix consisting of a single element and conclude that it is reducible, which is a contradiction. So $ \Lambda_{V_{k+1}} $ is unisolated. It means that not every column of $ \Lambda_{V_{k+1}} $ sums to one. Then by Corollary 3.3 of \cite{peterson1982leontief} and same arguments as in Claim 1 and Claim 2, if $ \text{det}(\mathbf{I}_{V_{k+1}}-\Lambda_{V_{k+1}})=0 $, then there must exist some $ v\in V_{k+1} $ such that $ 1-\lambda_{v,v}=0 $, which is a contradiction. So $ \text{det}(\mathbf{I}_{V_{k+1}}-\Lambda_{V_{k+1}})\neq 0 $, which implies that $ \text{Rank}(\mathbf{I}_{V_{k+1}}-\Lambda_{V_{k+1}})=|V_{k+1}| $.
\end{proof}

Given the block form (\ref{blockform}), Lemma \ref{lemma2} implies that
\[
\text{Rank}(\mathbf{I}-\Lambda)=\sum_{\ell=1}^{k+1} \text{Rank}(\mathbf{I}_{V_\ell}-\Lambda_{V_\ell})=|\cv|-k.
\] 
So $ (\mathbf{I}-\Lambda)\mathbf{x}=\mathbf{0} $ has $ k $ linearly independent solutions. Below we construct the $ k $ solutions.

For all every $ \ell=1,\ldots,k $, we consider the equation system $ (\mathbf{I}_{V_\ell}-\Lambda_{V_\ell})\mathbf{x}_{V_\ell}=\mathbf{0} $. Since $ \text{det}(\mathbf{I}_{V_\ell}-\Lambda_{V_\ell})=0 $, $ 1 $ is an eigenvalue of $ \Lambda_{V_\ell} $. Since $ \Lambda_{V_\ell} $ is irreducible, by Frobenius Theorem (see Section 6.8 of \cite{leonlinear}), $ 1 $ has a positive eigenvector $ \tilde{\mathbf{x}}_{V_\ell} $ that is a solution to $ (\mathbf{I}_{V_\ell}-\Lambda_{V_\ell})\mathbf{x}_{V_\ell}=\mathbf{0} $.
Recall that $ \Lambda $ can be written in the block form (\ref{blockform}). So
\[
\tilde{\mathbf{x}}_\ell=(\underbrace{\mathbf{0},\ldots,\mathbf{0}}_{\ell-1},\tilde{\mathbf{x}}_{V_\ell},\mathbf{0},\ldots,\mathbf{0})
\]
is a nonnegative solution to $ (\mathbf{I}-\Lambda)\mathbf{x}=\mathbf{0} $. It is clear that the $ k $ solutions $ \tilde{\mathbf{x}}_1,\tilde{\mathbf{x}}_2,\ldots,\tilde{\mathbf{x}}_k $ are linearly independent. Therefore, every solution to $ (\mathbf{I}-\Lambda)\mathbf{x}=\mathbf{0} $ is a linear combination of the $ k $ solutions. That is, there exist $ y_1,\ldots,y_k\in \mathbf{R} $ such that

\[
\mathbf{x}=y_1 \begin{bmatrix}
\tilde{\mathbf{x}}_{V_1} \\ \mathbf{0} \\ \mathbf{0} \\ \vdots \\ \mathbf{0} \\
\end{bmatrix}
+y_2 \begin{bmatrix}
\mathbf{0} \\\tilde{\mathbf{x}}_{V_2} \\ \mathbf{0} \\\vdots \\ \mathbf{0} \\
\end{bmatrix}
+\cdots
+y_k \begin{bmatrix}
\mathbf{0}  \\\vdots \\ \mathbf{0} \\\tilde{\mathbf{x}}_{V_k} \\ \mathbf{0}  \\
\end{bmatrix}.
\]
\medskip

For every $ \ell=1,\ldots,k $, define $ y^*_\ell=\min_{v\in V_\ell}\dfrac{q_v}{\tilde{x}_v} $.
Then, $ \mathbf{x}^*=\sum_{\ell=1}^k y^*_\ell \tilde{\mathbf{x}}_\ell $
is the maximum solution that satisfies the constraint $ \bx\le \mathbf{q} $.

This finishes the proof of Lemma \ref{thm:existence}.

\bigskip

Lemma \ref{lemma3} proves that  $ \{V_\ell\}_{1\le \ell \le k} $ is the set of absorbing sets in the graph $ (\mathcal{V},\mathcal{E}) $.

\begin{lemma}\label{lemma3}
A subset of nodes $ V\subseteq \cv $ is an absorbing set in the graph $ (\mathcal{V},\mathcal{E}) $ if and only if $ \Lambda_{V} $ is irreducible and isolated.
\end{lemma}
\begin{proof}
(Only if) Let $  V $ be an absorbing set. Because $ V $ has no outgoing edges, for every $ u\in V $ and $ v\in \cv\backslash V $, $ \lambda_{v,u}=0 $. It means that $ \Lambda_{V} $ is isolated. Because $ V $ is inside connected, for any strict subset $ V' $ of $ V $, there exist $ v\in V\backslash V' $ and $ u\in V' $ such that $ \lambda_{v,u}>0 $. It means that $ \Lambda_{V} $ is irreducible.

(If) Suppose $ \Lambda_{V} $ is irreducible and isolated. Being isolated directly means that $ V $ has no outgoing edges. Suppose there exist two nodes $ u,v\in V $ such that there is no directed path from $ u $ to $ v $. Let $ V_1 $ be the subset of $ V $ such that there is a directed path from $ u $ to every $ u'\in V_1 $. Let $ V_2 $ be the set of remaining nodes in $ V $. In particular, $ u\in V_1 $ and $ v\in V_2  $. Then there must be no directed path from every node in $ V_1 $ to every node in $ V_2 $; otherwise, there would be a directed path from $ u $ to every node in $ V_2 $, which contradicts the definition of $ V_2 $. It means that $ \lambda_{v,u}=0 $ for all $ v\in V_2 $ and $ u\in V_1 $. So $ \Lambda_{V_1} $ is isolated in $ \Lambda_{V} $, which contradicts that $ \Lambda_{V} $ is irreducible. Therefore, $ V $ must be inside connected.
\end{proof}

\section{Two ideas of solving Example \ref{example:2} through clearing cycles}\label{appendix:possible:idea}

This appendix discusses two ideas of solving Example \ref{example:2} through clearing cycles. Both ideas appear intuitive in some respects, but their found allocations do not satisfy desirable fairness.

\subsection{Clearing all cycles equally}

Our first idea is treating all cycles in a generated graph equally by trading an equal amount of objects in all cycles. Formally, let $ (\mathcal{V},\mathcal{E},\w) $ be a generated graph at any step of a mechanism in which $ \mathcal{V} $ is the set of nodes, $ \mathcal{E} $ is the set of directed edges, and $ \w $ is the endowments of agents. For every object-to-agent edge $ o\rightarrow i $, denote by $ n_{o\rightarrow i} $ the number of cycles involving the edge. Then, at most an amount $ \frac{\w_{i,o}}{n_{o\rightarrow i}} $ of object $ o $ can be traded in each of the $ n_{o\rightarrow i} $ cycles. So, in the graph, we trade an amount
\[
\min_{(o\rightarrow i)\in \mathcal{E}:\ n_{o\rightarrow i}>0} \dfrac{\w_{i,o}}{n_{o\rightarrow i}}
\] 
of objects in every cycle. This defines a mechanism. The mechanism is intuitively fair because the cycles at every step are treated equally. However, if we apply this mechanism to Example \ref{example:2}, we will find the allocation in \autoref{table:firstidea:allocation}, which violates EENE. Specifically, at step one, in \autoref{figure:example2}, we trade an amount $ 1/4 $ of objects in every cycle. So, $ 1 $ obtains $ 1/2c $, $ 2 $ obtains $ 1/4d $, $ 3 $ obtains $ 1/2d $, $ 4 $ obtains $ 3/4a $, and $ 5 $ obtains $ 1/4c $. The remaining steps are straightforward because the generated graphs are simple. At step two, $ 3 $ and $ 5 $ obtain $ 1/4c $. At step three, $ 2 $ obtains $ 1/4a $ and $ 5 $ obtains $ 1/2e $. At step four, $ 1 $ and $2 $ obtain $ 1/2b $, and $ 3 $ and $4 $ obtain $ 1/4e $. This allocation is sd-efficient but violates EENE for $ 1 $ and $ 2 $.

\subsection{Clearing non-redundant cycles}

Our second idea is to select non-redundant cycles to clear, which generalizes the idea of \cite{Kesten2009} to define a probabilistic version of TTC in the house allocation model. When a long cycle nests several short cycles, the long cycle might be viewed as redundant, because every agent in the long cycle can obtain the object she demands by clearing the short cycles. Formally, we call a cycle \textit{redundant} if it nests several shorter cycles and every node in the long cycle is involved in one of the shorter cycles. Otherwise, it is \textit{non-redundant}.

To define a mechanism, we need to specify the amount of objects traded in every non-redundant cycle. In every cycle, the amount of objects that can be traded is no more than the amounts of the objects owned by the agents in the cycle. If several cycles share an object-to-agent edge and the amount of the object owned by the agent is not enough to satisfy all of the cycles, we need to decide how to divide that amount among the several cycles.\footnote{Because agents' endowments are no more than their demands, there is no problem for clearing cycles if several cycles share an agent-to-object edge.}
However, this does not happen in Example \ref{example:2}. In \autoref{figure:example2}, cycle5 is redundant in the presence of cycle1 and cycle3. So, if we clear the four non-redundant cycles, the amount of objects traded in the four cycles will be $ 1/2 $, $ 1/2 $, $ 1/4 $, and $ 1/2 $,  respectively. After clearing the four cycles, $ 1 $ obtains $ 1/4c $, $ 2 $ obtains $ 1/2d $, $ 3 $ obtains $ 1/2d $, $ 4 $ obtains $ 3/4a $, and $ 5 $ obtains $ 1/2c $. The remaining steps are straightforward because all generated cycles are of the form $ i\rightarrow o\rightarrow i $. At step two, $ 3 $ obtains $ 1/4c $. At step three, $ 1 $ obtains $ 1/4a $ and $ 5 $ obtains $ 1/2e $. At step four, $ 1$ and $2 $ obtain $ 1/2b $, and $ 3 $ and $4 $ obtain $ 1/4e $. So, the mechanism will find the allocation in \autoref{table:secondidea:allocation}, which is sd-efficient but violates EENE for $ 1 $ and $ 2 $.

\begin{table}[!htb]
\centering
\scalebox{.9}{
	\begin{subtable}{.35\linewidth}
		\centering
		\begin{tabular}[c]{c|ccccc}
			& $a$ & $b$ & $c$ & $d$ & $e$ \\\hline
			$1$ &  & $1/2$ & $1/2$ &  &  \\
			$2$ & $1/4$ & $1/2$ &  & $1/4$ &  \\
			$3$ &  &  &  & $3/4$ & $1/4$ \\
			$4$ & $3/4$ &  &  &  & $1/4$ \\
			$5$ &  &  & $1/2$ &  & $1/2$ 
		\end{tabular}
		\subcaption{Outcome of the first mechanism}\label{table:firstidea:allocation}
	\end{subtable}
	\quad\quad
	\begin{subtable}{.4\linewidth}
		\centering
		\begin{tabular}[c]{c|ccccc}
			& $a$ & $b$ & $c$ & $d$ & $e$ \\\hline
			$1$ & $1/4$ & $1/2$ & $1/4$ &  &  \\
			$2$ &  & $1/2$ & & $1/2$ &  \\
			$3$ &  &  & $1/4$ & $1/2$ & $1/4$ \\
			$4$ & $3/4$ &  &  &  & $1/4$ \\
			$5$ &  &  & $1/2$ &  & $1/2$
		\end{tabular}
		\subcaption{Outcome of the second mechanism}\label{table:secondidea:allocation}
	\end{subtable}
}
\caption{}\label{table:xx}
\end{table}

\end{document}